\documentclass[journal,10pt]{IEEEtran}
\makeatletter
\def\endthebibliography{% 
    \def\@noitemerr{\@latex@warning{Empty `thebibliography' environment}}%
    \endlist
}
\makeatother

\usepackage{cite}

\usepackage[pdftex]{graphicx}
\usepackage[caption=false,font=footnotesize]{subfig}

\usepackage{amsmath}
\usepackage{mathtools, cuted}
\usepackage{amssymb}
\usepackage{bm}
\usepackage{mathrsfs}
\usepackage{breqn}

\usepackage{pifont}
\usepackage{xcolor}
\usepackage{url}
\usepackage{lettrine}
\usepackage{lipsum}
\usepackage{siunitx}
\usepackage{soul}
\usepackage{array}
\usepackage[inline]{enumitem}
\usepackage{epsfig}
\usepackage{filecontents}

\usepackage{algpseudocode}
\usepackage{algorithm, tabularx}
\usepackage{multirow}
\newcolumntype{L}[1]{>{\raggedright\let\newline\\\arraybackslash\hspace{0pt}}m{#1}}
\newcolumntype{C}[1]{>{\centering\let\newline\\\arraybackslash\hspace{0pt}}m{#1}}
\newcolumntype{R}[1]{>{\raggedleft\let\newline\\\arraybackslash\hspace{0pt}}m{#1}}
\newlength{\maxwidth}

\makeatletter
\newcommand{\multiline}[1]{%
	\begin{tabularx}{\dimexpr\linewidth-\ALG@thistlm}[t]{@{}X@{}}
		#1
	\end{tabularx}
}
\makeatother
\algdef{SE}[SUBALG]{Indent}{EndIndent}{}{{\algorithmicend\ }}
\algtext*{Indent}
\algtext*{EndIndent}

\usepackage{amsthm}
\newtheorem{theorem}{Theorem}

\theoremstyle{remark}

\newtheorem{remark}{Remark}

\DeclareMathOperator*{\argmax}{\arg\max}
\begin{document}

\title{{Fluid-Active-RIS (FARIS): A New Smart-Radio Paradigm with Joint Port Selection and\\Active Reflection Design}}

\author{Hong-Bae Jeon,~\IEEEmembership{Member,~IEEE}%, Kai-Kit Wong,~\IEEEmembership{Fellow,~IEEE}, and Chan-Byoung Chae,~\IEEEmembership{Fellow,~IEEE}%
\thanks{This work was supported by Hankuk University of Foreign Studies Research Fund of 2026. \textit{(Corresponding Author: Hong-Bae Jeon.)}}%
\thanks{H.-B. Jeon is with the Department of Information Communications Engineering, Hankuk University of Foreign Studies, Yong-in, 17035, Korea (e-mail: hongbae08@hufs.ac.kr).}%
%	\thanks{K.-K. Wong is with the Department of Electronic and Electrical Engineering, University College London, WC1E 6BT London, U.K., and also with the Yonsei Frontier Laboratory, Yonsei University, Seoul 03722, Korea (e-mail: kai-kit.wong@ucl.ac.uk).}
%\thanks{C.-B. Chae is with the School of Integrated Technology, Yonsei University, Seoul 03722, Korea (e-mail: cbchae@yonsei.ac.kr).}%
}

\maketitle

\begin{abstract}
In this paper, we introduce a new wireless paradigm termed \textit{fluid-active reconfigurable intelligent surface (FARIS)} that combines fluid-based port repositioning with per-element active amplification to enhance the performance of 6G networks. To realistically characterize the hardware operation, we first develop a circuit-level abstraction of the FARIS architecture and establish a practical power consumption model that captures both the logical control/switching power of candidate ports and the direct current (DC) bias power required for active reflection. Based on this model, we establish the FARIS signal model and formulate a corresponding ergodic-rate maximization problem that jointly optimizes the active amplification-reflection vector and the discrete selection of fluid-active elements under practical hardware constraints. The problem is addressed via an alternating optimization (AO) framework, which progressively improves the rate. Complexity and convergence analyses that follow furnish deeper insight into the algorithmic operation and performance enhancement. Numerical results confirm that the proposed FARIS with the AO framework consistently outperforms conventional baselines, delivering higher rates across diverse environments, often even when using fewer active elements or a smaller physical aperture.
\end{abstract}

\begin{IEEEkeywords}
Fluid active reconfigurable intelligent surface (FARIS), ergodic rate, alternating optimization.
\end{IEEEkeywords}

\IEEEpeerreviewmaketitle
\section{Introduction}
\label{sec:intro}
\lettrine{R}{ecently}, reconfigurable intelligent surfaces (RISs) have emerged as a key technology for sixth-generation (6G) wireless networks, primarily due to their capability to passively reconfigure the radio environment with minimal power consumption~\cite{RIST, RISnearmag}. An RIS typically comprises a large array of programmable, low-cost metasurface elements whose tunable reflection characteristics enable fine-grained manipulation of electromagnetic waves~\cite{holotut}. By coordinating the phase shifts across these elements, RISs can intentionally steer, focus, or scatter incident signals, thereby improving coverage, link reliability, spectral efficiency, and even localization performance across diverse wireless scenarios~\cite{dsRIS, nfRIS, HBRIS, HBRIS22}. Although RIS is recognized for its cost- and energy-efficiency in boosting wireless coverage and rate, its practical deployment remains challenging. One critical issue is that since conventional RIS architectures rely solely on passive reflecting elements, it inevitably suffers from severe multiplicative fading, as their end-to-end path loss (PL) is the product of the Tx-RIS and RIS-Rx channel losses. This fundamental limitation significantly degrades performance in practical deployments~\cite{DF, stardf}. 

To address this fundamental physical bottleneck, recent research has introduced active-RIS (ARIS) as a new architecture for wireless systems~\cite{aristut, aris1}. Unlike conventional passive-RIS (PRIS) that merely reflects incident signals, ARIS incorporates reflection-type amplifiers into their elements, enabling them to amplify the reflected signals. Although this requires additional power consumption, the resulting gain can effectively compensate for the severe double-reflection path loss, making ARIS a promising solution to the ``multiplicative fading'' problem. Several studies have further demonstrated the performance gains of ARIS in various wireless settings, including joint beamforming optimization with sub-connected architectures~\cite{aris2}, signal-to-noise-ratio (SNR)-oriented analyses under identical power budget with PRIS~\cite{aris4, aris5}, and average sum-rate maximization~\cite{aris8} and power-minimizing~\cite{aris7} frameworks under partial channel-state-information (CSI) conditions, highlighting its potential to fundamentally enhance coverage, reliability, and spectral efficiency.

Another critical issue is that conventional RIS architectures rely on a fixed and finitely quantized phase-control structure, which fundamentally limits their ability to realize the truly smart and highly adaptive wireless environment they were originally envisioned for~\cite{renzosmart}. As a result, practical RIS implementations often fall short of achieving environment-level reconfigurability, undermining the very motivation of deploying low-cost surfaces to offload complexity from base stations (BSs)~\cite{alexsmartris}. Furthermore, this fixed structure is also unfavorable for achieving high degree-of-freedom (DoF) in RIS-aided channels~\cite{RISvtm, nomaris}, as substantial DoF gains generally necessitate scaling up the RIS~\cite{ristit, jsacdofris, irris}, inevitably leading to increased training overhead and greater implementation complexity~\cite{rismeasure, rismeasure22}. Thus, the rigid structure of conventional RIS not only limits its adaptability to smart environments but also fundamentally restricts scalable DoF enhancement, which additionally limits its ability to effectively mitigate the aforementioned multiplicative-fading effect.

In this context, the concept of fluid antenna system (FAS) has emerged as a compelling solution~\cite{FAS, fluidtut}. Unlike conventional fixed-structure antennas, FAS offers shape- and position-reconfigurability enabled by flexible materials and architectures~\cite{metaant, pixfas}, thereby introducing additional physical-layer DoFs that can be exploited for performance enhancement~\cite{fasopdg}. Since the FAS is first introduced~\cite{FAS}, several studies have been made to understand the various performance limits of FAS~\cite{fasopdg, perlimfas, FASqout}, enable channel estimation schemes for FAS~\cite{fasmm, fasover}, and apply with several other applications including RIS~\cite{FASRIS, fasrisper}, integrated-sensing-and-communications (ISAC)~\cite{fasisac, fasisac22}, and direction-of-arrival (DoA) estimations~\cite{fasdoa, fasdoa22}. Taken together, these studies establish FAS as a promising paradigm that surpasses conventional antenna systems by unlocking higher spatial diversity, and by offering a versatile platform that can be seamlessly integrated into diverse 6G wireless functionalities.

Inspired by the position-reconfigurability of FAS, recent work has extended this concept to RIS through the notion of fluid-RIS (FRIS)~\cite{FRISlook, FRISmag}, where the RIS hosts a movable ``fluid'' elements capable of dynamically changing its position and applying optimized phase shifts~\cite{FRISsec}. By allowing a small number of mobile elements to traverse a larger physical aperture, FRIS effectively extracts spatial-domain DoFs without increasing the number of reflecting units~\cite{FRISpa}, offering a fundamentally different path toward scalable DoF enhancement under practical hardware constraints~\cite{FRISonoff}. While FRIS achieves meaningful performance gains through spatial-domain mobility, its reliance on passive reflection still leaves substantial headroom for further enhancement.

From this observation, we conceive that the advantages of spatial mobility and active reflection on RIS are not merely additive but fundamentally complementary. On the one hand, ARIS can compensate for severe cascaded path loss through signal amplification~\cite{aris1}; however, its amplification occurs over a fixed surface geometry, limiting the ability to concentrate amplification on spatially favorable locations. On the other hand, FRIS can adaptively select advantageous spatial positions across a larger aperture~\cite{FRISlook}, but passive reflection alone cannot sufficiently overcome the inherent multiplicative-fading penalty~\cite{DF}. Therefore, the key opportunity lies in \textit{a deeper integration where each selected element simultaneously possesses spatial mobility and controllable amplification.} In such a design, the surface can jointly optimize \textbf{where} the signal interaction occurs and \textbf{how strongly} the signal is reinforced, thereby creating a coupled spatial-amplitude design space that does not exist in conventional FRIS or static ARIS architectures. Importantly, such an architecture should not be interpreted as a simple system-level stacking of FRIS and ARIS functionalities. Rather, it requires a fundamentally new metasurface structure in which each reflecting element operates as a fluidic active unit, capable of dynamically repositioning within its designated subregion while simultaneously applying controllable amplification and phase adjustment. Through this intrinsic integration, spatial selection and signal reinforcement become jointly controllable design variables, enabling the surface to adapt both the spatial distribution and strength of reflected energy.

Motivated by this insight, we propose a \emph{fluid-active reconfigurable intelligent surface (FARIS)}, a new RIS paradigm that tightly integrates the spatial reconfigurability of FRIS with the signal amplification capability of ARIS. Unlike conventional FRIS or static ARIS architectures, FARIS equips each selected element with fluidic mobility together with active reflection functionality, allowing it to amplify and reflect incident signals while freely adjusting its position within the available aperture. This dual capability enables FARIS to simultaneously exploit position-adaptive spatial degrees of freedom and mitigate multiplicative fading through controllable active gain. As a result, FARIS provides a significantly stronger level of environment control compared with existing RIS architectures, offering a versatile platform for future 6G wireless networks. The main contributions are summarized as follows:

%This naturally raises the question of whether combining FRIS with the signal-amplification capability of ARIS could unlock even greater benefits. Importantly, however, our proposed architecture is not a mere system-level stacking of ARIS and FRIS mechanisms. Rather, it introduces a fundamentally new form of metasurface in which each active reflecting element functions as a fluidic unit, capable of providing controllable amplification while dynamically repositioning within its assigned subregion. In other words, this fluid-active behavior is intrinsically integrated into the surface itself.
%
%In particular, we introduce fluid-active-RIS (FARIS), a novel paradigm that integrates the signal-amplification functionality of ARIS with the spatial reconfigurability of FRIS. Unlike FRIS and static ARIS, FARIS equips the surface with a fluidic-active-element that can simultaneously amplify and reflect incoming signals while freely adjusting its position within a predefined aperture. This dual capability allows the surface to adapt its spatial energy distribution in real time, thereby enabling (i) position-adaptive spatial-DoF extraction and (ii) compensation of multiplicative fading through active gain. As a result, FARIS achieves a fundamentally stronger form of environment control, offering a powerful and versatile platform for 6G wireless network.
\begin{enumerate}
%\item We introduce the FARIS framework for a downlink setting and develop a unified formulation for joint port selection and amplification-phase design. By explicitly modeling the spatial correlation among selected FARIS ports and incorporating a practical ARIS reflection-power constraint, we cast an ergodic-rate maximization problem that simultaneously optimizes the amplification-reflection vector and the discrete configuration of fluid-reconfigurable ports.
\item We introduce a newly proposed FARIS framework for a downlink setting and develop a unified formulation for joint port selection and amplification-phase design. By explicitly modeling the spatial correlation among selected FARIS ports and incorporating a practical ARIS reflection-power constraint, we cast an ergodic-rate maximization problem that simultaneously optimizes the amplification-reflection vector and the discrete configuration of fluid-reconfigurable ports. In addition, based on a circuit-level modeling of the FARIS architecture, we establish a practical power consumption model that captures both the control-circuit power and the direct current (DC) bias power of the active reflection modules, which enables a realistic evaluation of the overall power consumption.
\item To address the mixed-integer non-convex nature of the problem, we develop an alternating-optimization (AO) framework for design on FARIS architecture. With the port configuration fixed, we employ sample average approximation (SAA) together with a quadratic fractional transform to convert the objective into a sequence of convex conic subproblems in a lifted matrix variable. When the resulting solution is not rank-one, Gaussian randomization and magnitude projection are applied. We additionally establish that this inner-loop procedure guarantees monotonic improvement of the SAA objective and convergence to a stationary point. For discrete port selection under a cardinality constraint on the selectable FARIS elements, given the current amplification-reflection vector, we propose a cardinality-constrained cross-entropy method (CEM) that draws Bernoulli selection samples and updates an independent Bernoulli distribution to fit the statistics of the high-rate elite subset. We derive a closed-form CEM update, showing that the optimal selection probabilities are parametrized by a single Lagrange multiplier and that the resulting mapping admits a unique solution enforcing the cardinality budget. The smoothed update is proven to monotonically increase the empirical log-likelihood, and yields a non-decreasing sequence of ergodic rates.
\item Building on these, we integrate the active-reflection design and the port selection into a unified AO framework. We prove that each outer iteration produces a non-decreasing sequence with finite limit by leveraging the monotonicity of the subproblems in AO framework. We conduct a detailed complexity analysis of each algorithmic block, where the resulting expressions characterize how the overall computational cost scales, thereby confirming the polynomial-time implementability of the FARIS with proposed AO framework.
\item Extensive simulations are performed over a wide range of system parameters, including transmit power, total number of FARIS elements and candidates, and normalized aperture. The results show that the proposed FARIS architecture with the AO framework exhibits fast convergence within a modest number of outer AO iterations and achieves near-optimal ergodic rate with only marginal loss compared to brute-force search (BFS), while drastically reducing computational burden. Moreover, FARIS consistently outperforms conventional baselines, achieving significantly higher rates even with fewer active elements or smaller apertures. In addition, the power consumption analysis demonstrates that FARIS requires only a slightly higher power than ARIS due to the candidate port control, while providing substantially improved rate performance. Moreover, compared to the amplify-and-forward (AF) relay architecture, FARIS avoids the excessive hardware power consumption associated with complex architectures. As a result, FARIS achieves a more favorable rate-power trade-off among the considered architectures.
\end{enumerate}
\textit{To the best of our knowledge, this is the first work to formalize FARIS as a new RIS paradigm that intrinsically couples fluid port reconfigurability with active reflection within a unified architecture. Supported by rigorous algorithmic design, detailed complexity characterization, and extensive simulations, our results show that FARIS is not merely a hybrid extension of existing RIS concepts, but a practically implementable and performance-enhancing paradigm that achieves near-optimal operation while consistently outperforming conventional FRIS and ARIS benchmarks.}

\begin{figure}[t]
	\begin{center}
		\includegraphics[width=0.8\columnwidth,keepaspectratio]%
		{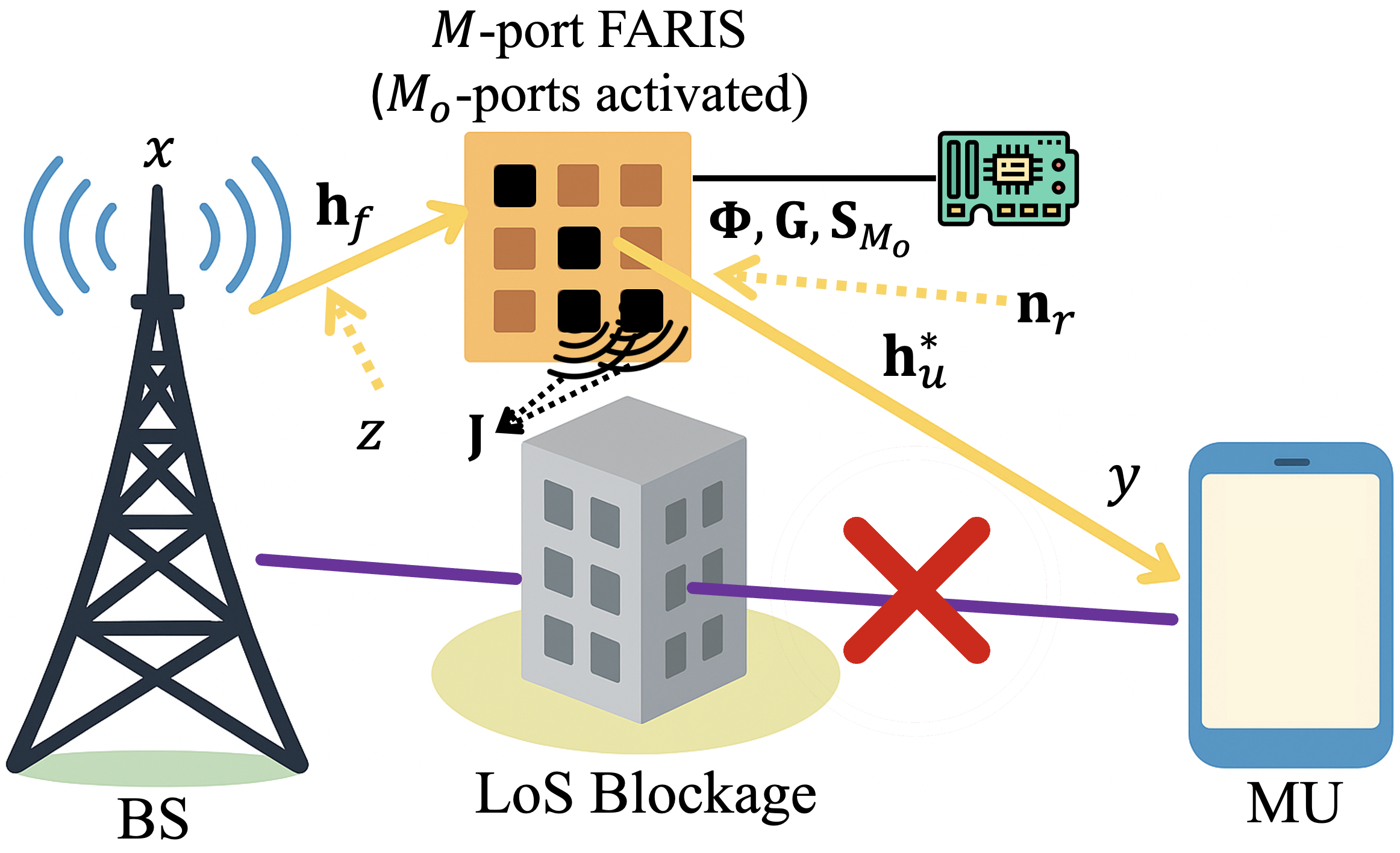}
		\caption{System model of FARIS-aided wireless network.}
		\label{fig_sm}
	\end{center}
\end{figure}
\section{System Model}
\label{sec:sys}
As illustrated in Fig.~\ref{fig_sm}, we examine a downlink wireless system assisted by FARIS. A BS equipped with a single fixed-position antenna (FPA) communicates with a mobile user (MU), also equipped with a single FPA, through the FARIS. The direct BS-MU link is assumed to be completely blocked by obstacles such as buildings. The FARIS consists of $M = M_x \times M_x$ reflective elements arranged over a surface of size $W_x\lambda \times W_x\lambda$, where $\lambda$ denotes the carrier wavelength and $W_x$ represents the normalized aperture of the FARIS relative to $\lambda$. The inter-element spacing is $d = \frac{W_x\lambda}{M_x}$, resulting in a spatial correlation matrix $\mathbf{J} \triangleq [J_{ij}] \in \mathbb{R}^{M \times M}$ modeled according to Jakes’ model. The spatial correlation between the $i$th and $j$th elements is given by~\cite{FRISonoff, FRISpa}
\begin{equation}
\label{corr}
J_{ij} = j_0\left( \frac{2\pi d_{ij}}{\lambda} \right),
\end{equation}
where $j_0(\cdot)$ denotes the zero-order spherical Bessel function of the first kind, and $d_{ij}=d\sqrt{ \left(\operatorname{mod}(i,M_x)-\operatorname{mod}(j,M_x)\right)^2+\left(\left\lfloor\frac{i}{M_x}\right\rfloor - \left\lfloor\frac{j}{M_x}\right\rfloor\right)^2}$ is the distance between the $i$ and $j$th elements. Following~\cite{FRISonoff, FRISmag}\footnote{Although the cited model is for FRIS, we adopt it here since FARIS follows the same hardware architecture, with the difference being the additional amplification module for each candidate.} and differing from~\cite{FRISlook}, each FARIS element can be viewed as a port of a fluid antenna structure. Each element operates in one of two states: ``on'' and ``off’’; In the ``on’’ state, the element is connected to the active reflection branch and actively interacts with the incident electromagnetic wave. Specifically, the impinging signal is processed through the controllable reflection circuitry, where its amplitude and phase can be adjusted according to the configured reflection coefficient. This enables the element to participate in the FARIS reflection process and contribute to shaping and amplifiying the composite reflected wave toward the desired direction. In contrast, in the ``off’’ state, the element is terminated by a matched load through the switching network. In this case, the incident electromagnetic energy is absorbed rather than re-radiated, effectively isolating the element from the reflection process and preventing any modification of the impinging signal. Consequently, the element does not contribute to the reflected field and only incurs the baseline control-circuit power consumption. The schematic illustration of the FARIS architecture is given in Fig.~\ref{fig_sche}, which serves as the basis for the power consumption model developed in the subsequent analysis.
\begin{figure}[t]
	\begin{center}
		\includegraphics[width=0.95\columnwidth,keepaspectratio]%
		{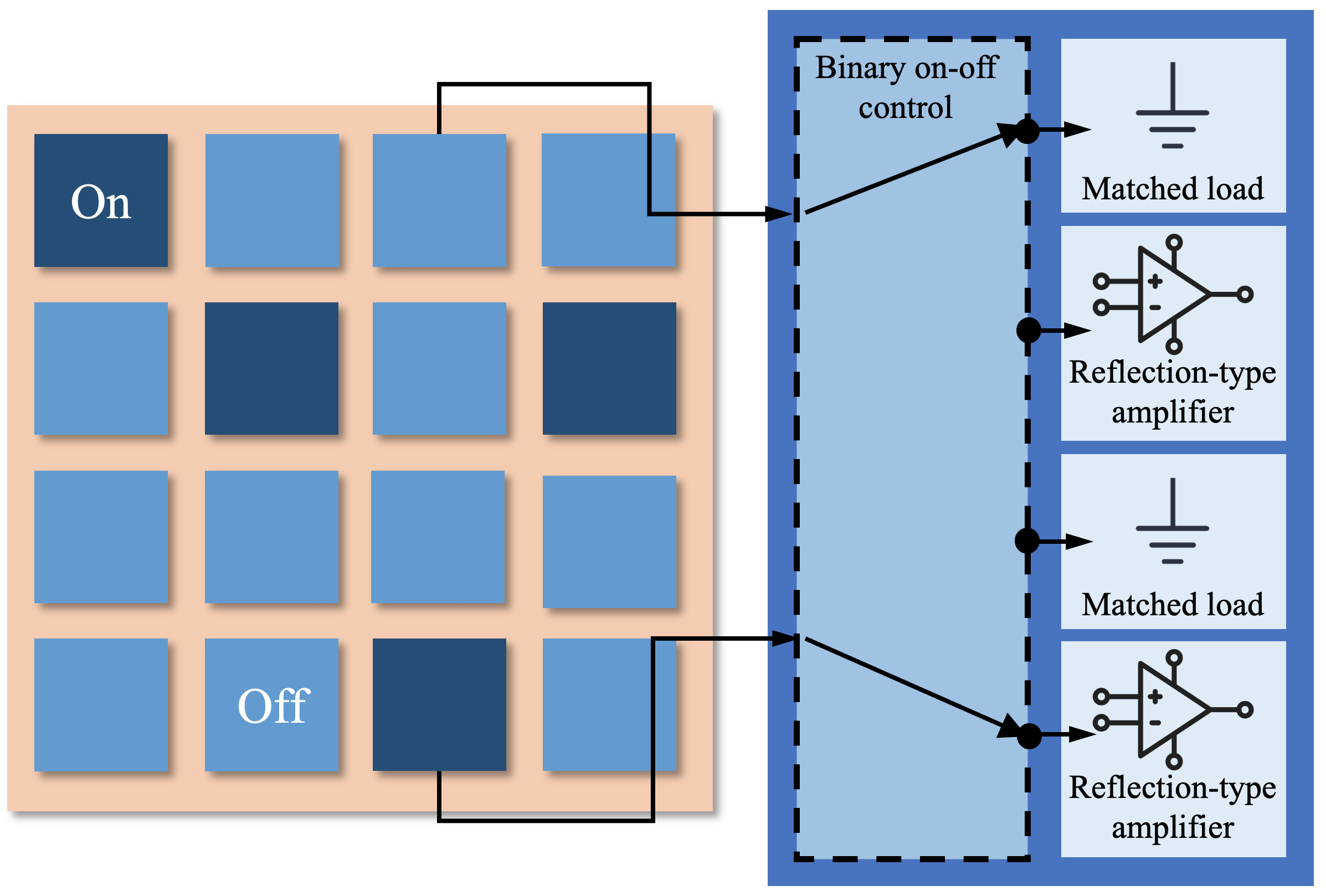}
		\caption{Schematic illustration of the FARIS architecture with incorporating on-off control circuits and reflection-type amplifiers for the selected elements.}
		\label{fig_sche}
	\end{center}
\end{figure}

%Each element operates in one of two states: ``on'' and ``off''. In the ``on'' state, the element actively interacts with incoming electromagnetic waves, modifying their amplitude and phase to realize the desired reflection. In the ``off'' state, the element is terminated by a matched load, effectively isolating it from the impinging wave and preventing any alteration of the signal.

The BS-FARIS link is assumed line-of-sight (LoS); $\tilde{\mathbf h}_f \in \mathbb C^M$, where as noted in~\cite{nearris}, it occurs when the BS-FRIS distance is shorter than the advanced MIMO Rayleigh distance (MIMO-ARD), represented as $D^{\text{RB}} \triangleq \lambda M$. The FARIS-MU link experiences Rician fading:
\begin{equation}
\label{eq:hu}
\mathbf{h}_u = \sqrt{\frac{K}{K+1}} \tilde{\mathbf{h}}_u + \sqrt{\frac{1}{K+1}} \hat{\mathbf{h}}_u\in\mathbb{C}^{M\times 1},
\end{equation}
where $\tilde{\mathbf h}_u$ and $\hat{\mathbf{h}}_u\sim \mathcal {CN}(0, \mathbf I_M)$ are the deterministic LoS and the non-LoS (NLoS) component, respectively, where $\mathbf I_M$ is an $M\times M$ identity matrix.

Out of the total $M$ candidate ports, only $M_o$ are selected. This selection process is represented as
\begin{equation}
\label{eq:Sm}
\mathbf{S}_{M_o}^{\mathrm{T}} \triangleq [\mathbf{e}_{i_1}  \cdots \mathbf{e}_{i_{M_o}}] \in \mathbb{R}^{M \times M_o}~(\forall i_n \in \{1, \cdots, M\}),
\end{equation}
where $\mathbf{e}_i$ denotes the $i$th canonical basis vector (i.e., the $i$th column of $\mathbf{I}_M$). This construction implies that only the ports indexed by $\{i_n\}_{n=1}^{M_o}$ are active on the FARIS. To ensure unique selection, all columns of $\mathbf{S}_{M_o}^{\mathrm{T}}$ must be mutually distinct, preventing redundant port selection.

The functionality of $\mathbf{S}_{M_o}$-utilized FARIS ports can be represented by a phase-shift $\boldsymbol\Phi$ and amplifier matrix $\mathbf G$:
\begin{equation}
\begin{aligned}
\label{eq:phiG}
&\boldsymbol \Phi = \mathrm{diag}\left(e^{j\phi_1},\cdots,e^{j\phi_{M_o}}\right)~(\forall \phi_i \in [0,2\pi)), \\
&\mathbf G = \mathrm{diag}\left(g_1,\cdots,g_{M_o}\right)~(\forall g_i \in [0, g_{\max}]).
\end{aligned}
\end{equation}
The effective operator $\mathbf A_{\mathrm{act}}$ of FARIS is therefore
\begin{equation}
\label{eq:Aact}
\mathbf A_{\mathrm{act}} \triangleq \mathbf J^{\frac{1}{2}}\mathbf S_{M_o}^{\mathrm T} \mathrm{diag}(\mathbf v)\mathbf S_{M_o}\mathbf J^{\frac{1}{2}},
\end{equation}
where
\begin{equation}
\label{defvv}
\mathbf v= [v_1\cdots v_{M_o}]^{\mathrm T},~ v_i = g_i e^{j\phi_i}\leftrightarrow \mathrm{diag}(\mathbf v)=\boldsymbol \Phi \mathbf G
\end{equation}
denotes the concatenated amplification ($\{g_i\}$) and reflection ($\{\phi_i\}$) mechanism on FARIS, where the magnitude and phase of $v_i$ correspond each, respectively.

The BS transmits a unit-power symbol $x$ with power $P$, and noise at the FARIS and the MU are respectively given by $\mathbf n_r \sim \mathcal{CN}(\mathbf 0,\sigma_r^2\mathbf I_M), ~ z \sim \mathcal{CN}(0,\sigma_0^2)$. In summary, the signal $\mathbf s_{\mathrm{RIS}}$ radiated by FARIS is
\begin{equation}
\label{eq:sRIS}
\mathbf s_{\mathrm{RIS}} = \mathbf A_{\mathrm{act}}\left(\sqrt{P L_f}\tilde{\mathbf h}_f x + \mathbf n_r\right),
\end{equation}
where $L_f$ is the PL of the BS-FARIS link.
%The instantaneous radiated power $P_{\mathrm{RIS}}$ at FARIS is therefore
%\begin{equation}
%\label{eq:Pris}
%P_{\mathrm{RIS}} = P L_f \|\mathbf A_{\mathrm{act}}\tilde{\mathbf h}_f\|_2^2+ \sigma_r^2\mathrm{tr}\left(\mathbf A_{\mathrm{act}}\mathbf A_{\mathrm{act}}^*\right),
%\end{equation}
%XXXXX
%
%XXXXX
{The instantaneous radiated power $P_{\mathrm{RIS}}$ at FARIS is therefore
\begin{equation}
\label{eq:Pris}
P_{\mathrm{RIS}} = P L_f \|\mathbf A_{\mathrm{act}}\tilde{\mathbf h}_f\|_2^2+ \sigma_r^2\mathrm{tr}\left(\mathbf A_{\mathrm{act}}\mathbf A_{\mathrm{act}}^*\right),
\end{equation}
which accounts for both the amplified incident signal component and the amplified thermal noise emitted by the FARIS. In addition to $P_{\mathrm{RIS}}$, the FARIS operation also incurs the per-element hardware power consumption due to the selected fluid-active ports. Specifically, following the structure and circuit model in Fig.~\ref{fig_sche}, the circuit power consumption of the FARIS is composed of two components: i) $P_c$: the logical control and switching power consumed by each candidate element, and ii) $P_{\mathrm{DC}}$: the DC bias power required for active reflection~\cite{aris5}. Since the logical control circuit exists for all $M$ candidates while the active reflection branch is given only for the selected $M_o$ FARIS ports, the total FARIS power consumption is modeled as
\begin{equation}
\label{eq:Pfaris_total}
P_{\mathrm{FARIS}}^{\mathrm{tot}}\triangleq MP_c+M_oP_{\mathrm{DC}}+\xi P_{\mathrm{RIS}},\footnote{Note that for ARIS, $M\leftarrow M_o$, and it returns to the conventional ARIS power consumption model in~\cite{aris5}.}
\end{equation}
where $\xi \triangleq \frac{1}{\upsilon}$ with $\upsilon\in(0,1]$ denoting the amplifier efficiency. Therefore, under the total FARIS power budget $P_{\max,t}$, the feasible reflection design must satisfy $P_{\mathrm{FARIS}}^{\mathrm{tot}} \le P_{\max,t}$, which is equivalent to:
\begin{equation}
\label{eq:Pris_constraint}
P_{\mathrm{RIS}}\le\upsilon\big(P_{\max,t}-MP_c-M_oP_{\mathrm{DC}}\big)\triangleq P_{\max}.
\end{equation}}
The received signal $y$ at MU is
\begin{equation}
\label{eq:rx}
\begin{aligned}
y&=\sqrt{L_u}\mathbf h_u^* \mathbf s_{\mathrm{RIS}}+z\\
&= \sqrt{P L_f L_u}\mathbf h_u^*\mathbf A_{\mathrm{act}}\tilde{\mathbf h}_f x+ \sqrt{L_u}\mathbf h_u^*\mathbf A_{\mathrm{act}}\mathbf n_r + z,
\end{aligned}
\end{equation}
where $L_u$ denotes the PL of the FARIS-MU link. From~\eqref{eq:rx}, the instantaneous signal-to-interference-noise-ratio (SINR) is
\begin{equation}
\label{eq:SINR_def}
\gamma =\frac{P L_f L_u \big|\mathbf h_u^*\mathbf A_{\mathrm{act}}\tilde{\mathbf h}_f\big|^2}{\sigma_0^2 + \mathbf h_u^*\mathbf B \mathbf h_u},
\end{equation}
where $\mathbf B \triangleq L_u\sigma_r^2 \mathbf A_{\mathrm{act}}\mathbf A_{\mathrm{act}}^*$, and the ergodic rate $\bar R$ is
\begin{equation}
\label{eq:ergodic_rate}
\bar R \triangleq \mathbb E_{\mathbf h_u}\left[\log_2(1+\gamma)\right].
\end{equation}
Hence, the overall design problem of the FARIS-enabled system can be formulated as
\begin{equation}
\label{prob:maxrate}
\max_{\mathbf v,\mathbf S_{M_o}} ~  \bar R~\text{s.t.}~\eqref{eq:Pris_constraint},~\forall|v_i|\le g_{\max},~\mathbf S_{M_o}\in\{0,1\}^{M_o\times M}.
\end{equation}
Note that since $\mathrm{diag}(\mathbf v)=\boldsymbol{\Phi}\mathbf G$, we solve~\eqref{prob:maxrate} and decompose $\{v_i\}$ into magnitude and phase, which corresponds to $\{g_i\}$ and $\{e^{j\phi_i}\}(\leftrightarrow\{\phi_i\})$, respectively.
\section{Proposed AO Framework}
\label{sec:AO}
To jointly optimize $\mathbf v ( \leftrightarrow (\boldsymbol\Phi,\mathbf G))$ and $\mathbf S_{M_o}$, we adopt an AO approach.
\subsection{Optimizing $\mathbf v ( \leftrightarrow (\boldsymbol\Phi,\mathbf G))$ for Fixed $\mathbf S_{M_o}$}
\label{subsec:Vopt}
When $\mathbf S_{M_o}$ is fixed, define the lifted variable $\mathbf V$ as
\begin{equation}
\label{defvv2}
\mathbf V \triangleq \mathbf v\mathbf v^* \succeq \mathbf 0~(\mathrm{rank}\mathbf V=1).
\end{equation}
We approximate~\eqref{eq:ergodic_rate} via the SAA using $S$ independent and identically distributed (i.i.d.) samples $\{\mathbf h^{(s)}\}_{s=1}^S$ of $\mathbf h_u$:
\begin{equation}
\label{eq:rate_SAA}
\bar R_S \triangleq \frac{1}{S}\sum_{s=1}^S \log_2\left(1+\underbrace{\frac{P L_f L_u \big|\mathbf h^{(s)*}\mathbf A_{\mathrm{act}}\tilde{\mathbf h}_f\big|^2}{\sigma_0^2 + \mathbf h^{(s)*}\mathbf B \mathbf h^{(s)}}}_{\triangleq\gamma_s}\right).
\end{equation}
To transform the problem with respect to $\mathbf V$, we introduce the following theorem:
\begin{theorem}
\label{thm:reformulation}
Define
\begin{equation}
\label{eq:thm_buk}
\mathbf b \triangleq \mathbf S_{M_o}\mathbf J^{\frac{1}{2}}\tilde{\mathbf h}_f,~\mathbf u_s \triangleq \mathbf S_{M_o}\mathbf J^{\frac{1}{2}}\mathbf h^{(s)},~\mathbf K \triangleq \mathbf S_{M_o}\mathbf J\mathbf S_{M_o}^{\mathrm T},
\end{equation}
and
\begin{equation}
\label{eq:thm_asCs}
\mathbf a_s \triangleq \mathbf u_s \odot \overline{\mathbf b},~
\mathbf C_s \triangleq L_u\sigma_r^2 \big(\mathbf K \odot (\mathbf u_s\mathbf u_s^*)\big),
\end{equation}
where $\overline{\mathbf b}$ is the element-wise conjugate of $\mathbf b$. Then the following hold:
\begin{itemize}
\item $\gamma_s$ can be written as
\begin{equation}
\label{eq:thm_gams}
\gamma_s
=
\frac{P L_f L_u \mathbf v^* \mathbf a_s \mathbf a_s^* \mathbf v}
     {\sigma_0^2 + \mathbf v^* \mathbf C_s \mathbf v}
=
\frac{P L_f L_u \mathrm{tr}(\mathbf a_s\mathbf a_s^* \mathbf V)}
     {\sigma_0^2 + \mathrm{tr}(\mathbf C_s\mathbf V)}.
\end{equation}
\item $P_{\mathrm{RIS}}$ satisfies
\begin{equation}
\label{eq:thm_Pris}
P_{\mathrm{RIS}}=\mathbf v^* \big(P L_f \mathbf Q_1 + \sigma_r^2 \mathbf Q_2\big) \mathbf v=\mathrm{tr}(\mathbf F \mathbf V),
\end{equation}
where $\mathbf F \triangleq P L_f\mathbf Q_1 + \sigma_r^2\mathbf Q_2$ and
\begin{equation}
\label{eq:thm_Q1Q2}
\mathbf Q_1
\triangleq
\mathrm{diag}(\mathbf b^*)\mathbf K \mathrm{diag}(\mathbf b)=\mathbf K \odot (\overline{\mathbf b}\mathbf b^{\mathrm T}),~\mathbf Q_2 \triangleq \mathbf K \odot \mathbf K.
\end{equation}
\end{itemize}
\end{theorem}
\begin{proof}
See Appendix A.
\end{proof}
Hence by Theorem~\ref{thm:reformulation}, the active reflection design subproblem becomes
\begin{equation}
\label{prob:Vopt_QCQP}
\begin{aligned}
\max_{\mathbf V}~
& \frac{1}{S}\sum_{s=1}^S
\log_2\Big(1+\frac{P L_f L_u\mathrm{tr}(\mathbf a_s\mathbf a_s^*\mathbf V)}
{\sigma_0^2+\mathrm{tr}(\mathbf C_s\mathbf V)}\Big)\\
\text{s.t.}~
& \mathrm{tr}(\mathbf F\mathbf V)\le P_{\max},~V_{ii}\le g_{\max}^2~(i=1,\cdots,M_o),\\
& \mathrm{rank}\mathbf V=1,~\mathbf V\succeq 0,
\end{aligned}
\end{equation}
where $V_{ii}$ is the $ii$th component of $\mathbf V$. Problem~\eqref{prob:Vopt_QCQP} is still not convex due to the fractional objective and the rank-1 constraint. To handle the nonconvex fractional structure, we will drop the rank constraint and apply a fractional/quadratic transform to obtain a sequence of convex conic subproblems in $\mathbf V$, followed by Gaussian randomization~\cite{sdr} and magnitude projection if necessary.

We first apply the quadratic transform, also known as the fractional programming technique, to~\eqref{prob:Vopt_QCQP}. For each $s$, introduce an auxiliary variable $y_s \ge0$ and use the following transformation for $A\ge0$ and $B>0$~\cite{fracp}:
\begin{equation}
\label{eq:quad_transform_identity}
\frac{A}{B}= \max_{y\ge 0} \bigl(2y\sqrt{A}-y^2 B\bigr),
\end{equation}
where the unique maximizer is given by $y^\star=\frac{\sqrt{A}}{B}$. Equation~\eqref{eq:quad_transform_identity} converts a ratio into a concave quadratic form with respect to $y\ge0$ and yields a closed-form update for $y$ given the other variables, enabling alternating optimization for sum-of-log-fraction objectives. Applying the result to~\eqref{eq:thm_gams},~\eqref{eq:rate_SAA} is equivalently lower-bounded as
\begin{equation}
\label{eq:transformed_obj}
\begin{aligned}
&\bar R_S \ge \frac{1}{S}\sum_{s=1}^S \log_2(1+\xi_s), ~\\
&\xi_s \triangleq P L_f L_u\Big(2y_s\sqrt{\mathrm{tr}(\mathbf a_s\mathbf a_s^*\mathbf V)}-y_s^2(\sigma_0^2+\mathrm{tr}(\mathbf C_s\mathbf V))\Big)(\ge0),
\end{aligned}
\end{equation}
where the maximum $\xi_s=\gamma_s$ holds by substituting $y_s =\frac{\sqrt{\mathrm{tr}(\mathbf a_s\mathbf a_s^*\mathbf V)}}{\sigma_0^2 + \mathrm{tr}(\mathbf C_s\mathbf V)}$. We will therefore maximize $\frac{1}{S}\sum_{s=1}^S \log_2(1+\xi_s)$ as the lower-bound of objective in~\eqref{prob:Vopt_QCQP}, which is constructed by the following procedure:
\subsubsection{Update of $\mathbf V$ ($\leftrightarrow\mathbf v$) for fixed $\{y_s\}$}
Since $\sqrt{\mathrm{tr}(\mathbf a_s\mathbf a_s^*\mathbf V)}$ in~\eqref{eq:transformed_obj} is nonconvex with respect to $\mathbf V$, by introducing the slack variables $\{w_s\ge0\}_{s=1}^S$, $\sqrt{\mathrm{tr}(\mathbf a_s\mathbf a_s^*\mathbf V)}$ is handled by $w_s$ subject to the second-order cone (SOC) $\mathcal Q_r$~\cite{boyd}:
\begin{equation}
\label{socde}
\Big(w_s, \sqrt{\mathrm{tr}(\mathbf a_s\mathbf a_s^*\mathbf V)}\Big)\in\mathcal Q_r~\leftrightarrow~ w_s\le \sqrt{\mathrm{tr}(\mathbf a_s\mathbf a_s^*\mathbf V)},
\end{equation}
which is equivalent to stating that $w_s$ forms the hypograph of $\sqrt{\mathrm{tr}(\mathbf a_s\mathbf a_s^*\mathbf V)}$. Hence, for fixed $\{y_s\}$, the transformed subproblem becomes
\begin{equation}
\label{prob:V_SOC}
\begin{aligned}
\max_{\mathbf V,\{w_s, \xi_s\}} & 
~\frac{1}{S}\sum_{s=1}^S \log_2(1+\xi_s)\\
\text{s.t.}~
& \xi_s \le P L_f L_u\left(2y_s w_s - y_s^2\big(\sigma_0^2+\mathrm{tr}(\mathbf C_s\mathbf V)\big)\right),\\
& \Big(w_s, \sqrt{\mathrm{tr}(\mathbf a_s\mathbf a_s^*\mathbf V)}\Big)\in\mathcal Q_r,\\
& w_s\ge0,~\xi_s\ge0,~ \mathbf V\succeq0,\\
& \mathrm{tr}(\mathbf F\mathbf V)\le P_{\max},~V_{ii}\le g_{\max}^2.\\
& (i=1,\cdots,M_o,~s=1,\cdots,S).
\end{aligned}
\end{equation}
Problem~\eqref{prob:V_SOC} is a convex conic optimization problem with linear matrix inequality (LMI) and SOC constraints. Upon solving it, if the optimal $\mathbf V_0$ is rank-one, the corresponding $\mathbf v_0$ is directly obtained as its principal eigenvector satisfying $\mathbf V_0 = \mathbf v_0 \mathbf v_0^{*}$. Otherwise, $\mathbf v_0$ is extracted via Gaussian randomization, starting by computing the Cholesky factorization of $\mathbf V_0:~\mathbf V_0 = \mathbf L \mathbf L^{*}$~\cite{chol}, where $\mathbf L$ is a lower-triangular matrix. We then generate $N_{\mathrm{rand}}$ independent Gaussian realizations:
\begin{equation}
\label{eq:rand_vectors}
{\mathbf v}_n = \mathbf L \mathbf z_n~(n = 1,\cdots, N_{\mathrm{rand}}),
\end{equation}
where $\{\mathbf z_n \sim \mathcal{CN}(\mathbf 0, \mathbf I_{M_o})\}$ are i.i.d. standard circularly symmetric complex Gaussian random vectors. For each ${\mathbf v}_n$, we first enforce feasibility with respect to the hardware constraints. Specifically, we apply an element-wise magnitude projection to satisfy $\{V_{ii}\le g_{\max}^2\}$ in~\eqref{prob:V_SOC},
\begin{equation}
\label{eq:mag_projection_re}
\big[\mathbf v_{n}\big]_i \leftarrow \min\{g_{\max}, |[\mathbf v_{n}]_i|\}e^{j\angle [\mathbf v_{n}]_i}~(i=1,\cdots,M_o),
\end{equation}
followed by a global power scaling
\begin{equation}
\label{eq:power_scaling}
\mathbf v_n \leftarrow
\alpha_n \mathbf v_n,~\alpha_n \triangleq \min\left\{1,\sqrt{\frac{P_{\max}}{\mathbf v_n^{*}\mathbf F\mathbf v_n}}\right\},
\end{equation}
which guarantees $\mathbf v_n^{*}\mathbf F\mathbf v_n\le P_{\max}$ in~\eqref{prob:V_SOC}. Let $\mathbf v_n^{\mathrm{feas}}$ denote the resulting feasible vector and define the corresponding rank-one matrix $\mathbf V_n\triangleq \mathbf v_n^{\mathrm{feas}}\mathbf v_n^{\mathrm{feas}*}$. For each $\mathbf V_n$, we evaluate the objective in~\eqref{prob:Vopt_QCQP}, denoted by $\tilde R_n$, then select the best randomized candidate as
\begin{equation}
\label{eq:n_star_new}
n^\star = \argmax_{n\in\{1,\cdots,N_{\mathrm{rand}}\}} \tilde{R}_n,
\end{equation}
which corresponds to $\mathbf v_{n^\star}^{\mathrm{feas}}$.
\subsubsection{Update of $y_s$ for fixed $\mathbf V$}
For fixed $\mathbf V$, each $y_s$ admits the closed-form update
\begin{equation}
\label{eq:y_update_SOC}
y_s^{\star} = \frac{w_s}{\sigma_0^2+\mathrm{tr}(\mathbf C_s\mathbf V)}~(s=1,\cdots,S).
\end{equation}
By alternately updating $\mathbf V$ from~\eqref{prob:V_SOC} and $\{y_s\}$ from~\eqref{eq:y_update_SOC},
the objective is guaranteed to be non-decreasing until convergence, which will be shown later.
\begin{algorithm}[t]
\caption{AO Framework of $\mathbf v$ and $\{y_s\}$ for Fixed $\mathbf{S}_{M_o}$}
\label{alg:AO}
\begin{algorithmic}[1]
\State \textbf{Input:} $\{\mathbf h^{(s)}\}_{s=1}^S,~\epsilon_v,~\mathbf S_{M_o}$
\State \textbf{Initialize:} $\mathbf v^{[0]}(\leftrightarrow\mathbf V^{[0]})$ satisfying~\eqref{eq:Pris_constraint} and $\forall|v_i^{[0]}|<g_{\max}$, $\{y_s^{[0]}\},~t \leftarrow 0$
\State Evaluate $\bar R_S^{[0]}$ by~\eqref{eq:rate_SAA}
\Repeat
    \State Solve~\eqref{prob:V_SOC} to obtain $\mathbf V^{[t+1]}$
    \If{$\mathrm{rank}\mathbf V^{[t+1]} > 1$}
        \State Extract $\mathbf v^{[t+1]}$ via Gaussian randomization
    \Else
        \State $\mathbf v^{[t+1]} \leftarrow$ principal eigenvector of $\mathbf V^{[t+1]}$
    \EndIf
        \For{$s = 1, \cdots, S$}
        \State Update $y_s^{[t+1]}$ via~\eqref{eq:y_update_SOC} using $\mathbf V^{[t+1]}$
    \EndFor
    \State Evaluate $\bar R_S^{[t+1]}$ by~\eqref{eq:rate_SAA}
    \State $t \leftarrow t + 1$
\Until{$|\bar R_S^{[t]} - \bar R_S^{[t-1]}| < \epsilon_v$}
\State $\mathbf v^\dagger \leftarrow \mathbf v^{[t]}(\leftrightarrow\mathbf V^\dagger \leftarrow \mathbf V^{[t]})$
\State \textbf{Output:} $\mathbf v^\dagger(\leftrightarrow\mathbf V^\dagger)$
\end{algorithmic}
\end{algorithm}
The overall procedure is given by Algorithm~\ref{alg:AO}.
\begin{remark}
\label{r11}
In Algorithm~\ref{alg:AO}, we initialize $\mathbf v^{[0]}(\leftrightarrow\mathbf V^{[0]})$ by drawing i.i.d. $\theta_i\sim\mathcal U[0,2\pi)$, set $\bar{\mathbf v}=g_{\max}[e^{j\theta_1} \cdots e^{j\theta_{M_o}}]^{\mathrm T}$, and scale $\bar{\mathbf v}\rightarrow \mathbf v^{[0]}$ so that $\mathbf v^{[0]*} \mathbf F\mathbf v^{[0]}\le P_{\max}$. This guarantees feasibility of~\eqref{eq:Pris_constraint} and $\forall|v_i^{[0]}|\le g_{\max}$. Thereafter, we set $y_s^{[0]} = \frac{\sqrt{\mathrm{tr}(\mathbf a_s\mathbf a_s^* \mathbf V^{[0]})}}{\sigma_0^2 + \mathrm{tr}(\mathbf C_s \mathbf {V}^{[0]})}~(\forall s)$. 
\end{remark}
\subsubsection{Monotonicity and Convergence Analysis of Algorithm~\ref{alg:AO}}
To ensure that Algorithm~\ref{alg:AO} produces stable performance improvements, it is essential to verify that each iteration does not degrade $\bar{R}_S$. The following theorem establishes this monotonicity property.
\begin{theorem}
\label{th1}
In Algorithm~\ref{alg:AO}, $\{{\bar R_S^{[t]}}\}$ monotonically non-decreases.
\end{theorem}
\begin{proof}
See Appendix B.
\end{proof}
Furthermore, since the feasible set of $\mathbf V$ is compact due to the constraints $\mathrm{tr}(\mathbf F\mathbf V)\le P_{\max}$ and $V_{ii}\le g_{\max}^2~(i=1, \cdots, M_o)$ in~\eqref{prob:Vopt_QCQP}, $\{\bar R_S^{[t]}\}$ is bounded above. Therefore, $\{\bar R_S^{[t]}\}_{t\ge 0}$ converges to a finite limit.
\subsection{Optimizing $\mathbf S_{M_o}$ for fixed $\mathbf v$}
\label{subsec:stoch_sel}
To stochastically determine the active ports for given $\mathbf v$, each port $i$ is assigned an selection probability $p_i\in[0,1]$, forming the probability vector $\mathbf p=[p_1\cdots p_M]^{\mathrm T}$. Let $\boldsymbol\zeta=[\zeta_1\cdots\zeta_M]^{\mathrm T}\in\{0,1\}^M$ denote a random binary port-selection vector sampled as
\begin{equation}
\label{eq:bern}
\boldsymbol\zeta \sim \mathrm{Bernoulli}(\mathbf p),~\mathbb P(\boldsymbol\zeta;\mathbf p)=\prod_{i=1}^{M}p_i^{\zeta_i}(1-p_i)^{1-\zeta_i}.
\end{equation}
We propose to utilize an $M_o$-constrained CEM that fits an independent Bernoulli law $\boldsymbol\zeta\sim\mathrm{Bernoulli}(\mathbf p)$ to the elite set of high-performing port subsets while exactly enforcing the expected-cardinality budget $\sum_{i=1}^M p_i=M_o$~\cite{CEM, FRISonoff, FRISsec}.
\subsubsection{Ideal Target Distribution}
Let $\mathcal Z\subseteq\{0,1\}^M$ denote the set of all feasible binary port-selection vectors, and let $\bar R_S(\boldsymbol\zeta)$ be the objective in~\eqref{eq:rate_SAA} corresponding to $\boldsymbol\zeta\in\mathcal Z$. Define the elite threshold $\gamma_\rho$ as the $(1-\rho)$-quantile of $\bar R_S(\boldsymbol\zeta)$ under the current sampling distribution, so that a fraction $\rho$ of samples satisfy $\bar R_S(\boldsymbol\zeta)\ge\gamma_\rho$. The ideal target distribution $\mathbb P^\dagger(\boldsymbol\zeta)$ that perfectly concentrates on the elite set is then
\begin{equation}
\label{eq:idealP}
\mathbb P^\dagger(\boldsymbol\zeta)\propto
\begin{cases}
1~(\bar R_S(\boldsymbol\zeta)\ge\gamma_\rho)\\
0~(\text{otherwise}).
\end{cases}
\end{equation}
Intuitively, $\mathbb P^\dagger$ assigns nonzero probability only to high-performing (elite) selection patterns.

The CEM seeks $\mathbb P(\boldsymbol\zeta;\mathbf p)$ that approximates $\mathbb P^\dagger$ as closely as possible. Formally, this is achieved by minimizing the Kullback-Leibler (KL) divergence:
\begin{equation}
\label{eq:KL}
\begin{aligned}
\mathbf p^{\mathrm{new}}
&\triangleq\arg\min_{\mathbf p}
D_{\mathrm{KL}}\big(\mathbb P^\dagger \| \mathbb P(\cdot;\mathbf p)\big)\\
&=\arg\max_{\mathbf p}
\mathbb E_{\mathbb P^\dagger}
\big[\log \mathbb P(\boldsymbol\zeta;\mathbf p)\big].
\end{aligned}
\end{equation}
Since $\mathbb P^\dagger$ is unknown, we approximate the expectation in~\eqref{eq:KL} using the empirical distribution of the elite samples:
\begin{equation}
\label{eq:CE_emp}
\mathbb E_{\mathbb P^\dagger}\big[\log \mathbb P(\boldsymbol\zeta;\mathbf p)\big]
\approx
\frac{1}{N_e}\sum_{n\in\mathcal E}
\log \mathbb P(\boldsymbol\zeta_n;\mathbf p),
\end{equation}
where we draw $N_{\mathrm{mc}}$ i.i.d. samples $\{\boldsymbol\zeta_n\}_{n=1}^{N_{\mathrm{mc}}}$ from $\mathrm{Bernoulli}(\mathbf p)$ and evaluate $\{\bar R_S(\boldsymbol\zeta)\}_{n=1}^{N_{\mathrm{mc}}}$. Let $\mathcal E$ be the elite index set of the top-$\rho$ fraction, with size of $N_e=\lceil \rho N_{\mathrm{mc}}\rceil$ by definition.
%where $\mathcal E$ is the index set of elite samples and $N_e=|\mathcal E|$.
Maximizing~\eqref{eq:CE_emp} yields
\begin{equation}
\label{eq:CE_emp2}
\mathbf p^{\mathrm{new}}
=\arg\max_{\mathbf p\in[0,1]^M}
\frac{1}{N_e}\sum_{n\in\mathcal E}
\log \mathbb P(\boldsymbol\zeta_n;\mathbf p).
\end{equation}
which is exactly the empirical log-likelihood maximization used in the CEM update rule.
\subsubsection{Applying to $M_o$-constrained CEM with Iterations}
We now modify the CEM with the aforementioned cardinality budget. Herein at iteration $t$, draw $N_{\mathrm{mc}}$ i.i.d. samples $\{\boldsymbol\zeta_n^{[t]}\}_{n=1}^{N_{\mathrm{mc}}}$ from $\mathrm{Bernoulli}(\mathbf p^{[t]})$ and evaluate $\{\bar R_S(\boldsymbol\zeta_n^{[t]})\}_{n=1}^{N_{\mathrm{mc}}}$. Let $\mathcal E^{[t]}$ be the elite index set of the top-$\rho$ fraction, with size of $N_e=\lceil \rho N_{\mathrm{mc}}\rceil$ by definition. Then we can define the elite empirical average per coordinate $i$:
\begin{equation}
\label{eq:mu_def}
\mu_{i}^{[t]} \triangleq \frac{1}{N_e}\sum_{n\in \mathcal E^{[t]}} \zeta_{n,i}^{[t]},
\end{equation}
where $\zeta_{n,i}^{[t]}$ is the $i$th component of $\boldsymbol\zeta_n^{[t]}~(i=1, \cdots, M)$. %and $\mu_i^{[t]}\in (0, 1)$ holds by definition, and 
By substituting into~\eqref{eq:CE_emp2}, it becomes
\begin{equation}
\label{prob:CE_main}
\begin{aligned}
\max_{\mathbf p}~ &
\Phi(\mathbf p)\triangleq\sum_{i=1}^M\Big(\mu_i^{[t]}\log p_i + (1-\mu_i^{[t]})\log(1-p_i)\Big)\\
\text{s.t.}~ & \sum_{i=1}^M p_i = M_o,~\mathbf p\in(0,1)^M.
\end{aligned}
\end{equation}
Problem~\eqref{prob:CE_main} is strictly concave with an affine constraint, hence admits a unique global maximizer, as characterized in the following theorem.
\begin{theorem}
\label{lem:CEM_KKT}
Let $\nu\in\mathbb R$ be the Lagrange multiplier. The unique maximizer of~\eqref{prob:CE_main} satisfies, for each $i$,
\begin{equation}
\label{eq:quad_eq}
\underbrace{\nu p_i^2 - (\nu+1)p_i + \mu_i^{[t]}}_{\triangleq f(p_i,\nu)}=0,
\end{equation}
whose solution in $(0,1)$ is
\begin{equation}
\label{eq:p_of_nu}
p_i(\nu) =
\begin{cases}
\mu_i^{[t]} & (\nu=0)\\
\displaystyle\frac{(\nu+1)-\sqrt{(\nu+1)^2-4\nu\mu_i^{[t]}}}{2\nu} & (\nu\neq 0).
\end{cases}
\end{equation}
Moreover, the mapping $\nu\mapsto \sum_{i=1}^M p_i(\nu)$ is strictly decreasing and continuous on its domain. Hence there exists a unique $\nu^\dagger$ such that $\sum_i p_i(\nu^\dagger)=M_o$.
\end{theorem}
\begin{proof}
See Appendix C.
\end{proof}
Given $\nu^\dagger$, let $\mathbf p^{\mathrm{CE}}\triangleq\mathbf p(\nu^\dagger)$ the exact maximizer of CEM in~\eqref{prob:CE_main}. A standard smoothed update is then
\begin{equation}
\label{eq:smoothed_update}
\mathbf p^{[t+1]}=(1-\omega)\mathbf p^{[t]}+\omega\mathbf p^{\mathrm{CE}}~(\omega\in(0,1)),
\end{equation}
which preserves the expected ``$M_o$-sum'' budget exactly since $\sum_i p_i^{[t]}=\sum_i p_i(\nu^\dagger)=M_o$ and by linearity in~\eqref{eq:smoothed_update}, and improves $\Phi$ by following theorem.
\begin{theorem}
\label{prop:monotone}
$\mathbf p^{\mathrm{CE}}=\mathbf p(\nu^\dagger)$ strictly increases $\Phi$ unless $\mathbf p^{[t]}$ already solves~\eqref{prob:CE_main}. Moreover, for any $\omega\in(0,1]$, the smoothed update~\eqref{eq:smoothed_update} yields $\Phi(\mathbf p^{[t+1]})\ge \Phi(\mathbf p^{[t]})$.
\end{theorem}
\begin{proof}
Strict concavity and uniqueness of the maximizer imply $\Phi(\mathbf p^{\mathrm{CE}})>\Phi(\mathbf p^{[t]})$ unless optimality holds at $\mathbf p^{[t]}$. For the smoothed step, $\mathbf p^{[t+1]}=(1-\omega)\mathbf p^{[t]}+\omega\mathbf p^{\mathrm{CE}}$ remains feasible due to the strict concavity, and the concavity of $\Phi$ gives $\Phi(\mathbf p^{[t+1]})\ge (1-\omega)\Phi(\mathbf p^{[t]})+\omega\Phi(\mathbf p^{\mathrm{CE}})\ge \Phi(\mathbf p^{[t]})$, and the theorem follows.
\end{proof}
The overall procedure is summarized in Algorithm~\ref{alg:CEM_ports}, whose monotonicity and convergence is guaranteed by theories of CEM in~\cite{CEM}. After verifying that the proposed algorithm concentrates $\mathbf p^{[t]}$ to $\mathbb P^\dagger$, we finally choose $\widehat{\mathcal I}$ by selecting the top-$M_o$ indices of $\mathbf p^{[t]}$, and form the optimal $\mathbf{S}_{M_o}^\dagger$.
\begin{algorithm}[t]
\caption{CEM-Based Selection of $\mathbf S_{M_o}$ for Fixed $\mathbf{v}$}
\label{alg:CEM_ports}
\begin{algorithmic}[1]
\State \textbf{Input:} $M_o,~\rho,~N_{\mathrm{mc}},~\omega,~\epsilon_c,~\mathbf v$
\State \textbf{Initialize:} $\mathbf p^{[0]}\in(0,1)^M$ with $\sum_i p_i^{[0]}=M_o$ (e.g., uniform distribution), $t\leftarrow 0$
\Repeat
  \State Draw $\boldsymbol\zeta_n^{[t]}\sim\mathrm{Bernoulli}(\mathbf p^{[t]})~(n=1,\cdots,N_{\mathrm{mc}})$
  \State Compute $\{\bar R_S(\boldsymbol\zeta_n^{[t]})\}_{n=1}^{N_{\mathrm{mc}}}$ using $\mathbf v$
  \State $\mathcal E^{[t]}\leftarrow$ indices of top-$\rho$ samples by $\bar R_S$
  \State Compute $\{\mu_i^{[t]}\}_{i=1}^M$ by~\eqref{eq:mu_def}
  \State Find the unique $\nu^\dagger$ %s.t. $\sum_{i=1}^M p_i(\nu^\dagger)=M_o$ with $p_i(\nu)$
  from~\eqref{eq:p_of_nu}
  \State Update $\mathbf p^{[t+1]}$ by~\eqref{eq:smoothed_update}
  \State $t\leftarrow t+1$
\Until{$||\mathbf p^{[t+1]}-\mathbf p^{[t]}||<\epsilon_c$}
\State $\widehat{\mathcal I}\leftarrow \mathrm{top\text{-}}M_o(\mathbf p^{[t]})$ and form $\mathbf S_{M_o}^\dagger$ by $\widehat{\mathcal I}$
\State \textbf{Output:} $\mathbf S_{M_o}^\dagger$
\end{algorithmic}
\end{algorithm}
\begin{algorithm}[t]
\caption{Overall AO Framework}
\label{alg:overall_AO_compact}
\begin{algorithmic}[1]
\State \textbf{Initialize:} $\mathbf v^{(0)},~\mathbf S_{M_o}^{(0)},~t\leftarrow0$
\State Compute $\bar R_S^{(0)}$ by~\eqref{eq:rate_SAA}
\Repeat
\State Given $\mathbf S_{M_o}^{(t)}$, run Algorithm~\ref{alg:AO} to obtain $\mathbf v^{(t+1)}$
\State Given $\mathbf v^{(t+1)}$, run Algorithm~\ref{alg:CEM_ports} to obtain $\mathbf S_{M_o}^{(t+1)}$.
  \State $\bar R_S^{(t+1)} \leftarrow \frac{1}{S}\sum_{s=1}^S \log_2\big(1+\gamma_s(\mathbf v^{(t+1)},\mathbf S_{M_o}^{(t+1)})\big)$
  \State $t\leftarrow t+1$
\Until{$\big|\bar R_S^{(t+1)}-\bar R_S^{(t)}\big|<\epsilon_{\mathrm{out}}$}
\State \textbf{Output:} $\mathbf v^\star\leftarrow \mathbf v^{(t)},~\mathbf S_{M_o}^\star\leftarrow \mathbf S_{M_o}^{(t)},~\bar R_S^\star\leftarrow \bar R_S^{(t)}$
\end{algorithmic}
\end{algorithm}
\subsection{Overall Procedure}
The two subproblems in Section~\ref{subsec:Vopt} and~\ref{subsec:stoch_sel} are alternately solved until convergence:
\begin{equation}
\label{svssvs}
\cdots\rightarrow\mathbf S_{M_o}^{(t)} \rightarrow\mathbf V^{(t+1)} \rightarrow\mathbf S_{M_o}^{(t+1)}\rightarrow\cdots,
\end{equation}
where the overall procedure is given in Algorithm~\ref{alg:overall_AO_compact}. Note that $\mathbf v^{(0)}$ is initialized according to Remark~\ref{r11}, while $\mathbf S_{M_o}^{(0)}$ is generated by randomly selecting $M_o$ active ports from the $M$ available candidates. This proposed AO framework effectively decouples the continuous reflection design and the discrete port-selection, while ensuring monotonic improvement of the overall objective.
\subsection{Overall Convergence of Algorithm~\ref{alg:overall_AO_compact}}
We now combine the above results to discuss the convergence of the full AO procedure in Algorithm~\ref{alg:overall_AO_compact}, Let $\bar R_S^{(t)}$ denote the SAA objective value in~\eqref{eq:rate_SAA} evaluated at $(\mathbf v^{(t)},\mathbf S_{M_o}^{(t)})$ in Algorithm~\ref{alg:overall_AO_compact}. Then, for each outer iteration $t$:
\begin{itemize}
  \item Given $\mathbf S_{M_o}^{(t)}$, Algorithm~\ref{alg:AO} produces $\mathbf v^{(t+1)}$ such that
  \begin{equation}
  \label{mono1}
  \bar R_S(\mathbf v^{(t+1)},\mathbf S_{M_o}^{(t)})
  \ge
  \bar R_S(\mathbf v^{(t)},\mathbf S_{M_o}^{(t)}),
  \end{equation}
  due to the monotonicity property by Theorem~\ref{th1}
  \item Given $\mathbf v^{(t+1)}$ and sufficiently large $N_{\mathrm{mc}}$, Algorithm~\ref{alg:CEM_ports} yields $\mathbf S_{M_o}^{(t+1)}$ which guarantees~\cite{CEM}
  \begin{equation}
  \label{mono2}
  \bar R_S(\mathbf v^{(t+1)},\mathbf S_{M_o}^{(t+1)})
  \ge
  \bar R_S(\mathbf v^{(t+1)},\mathbf S_{M_o}^{(t)}).
  \end{equation}
\end{itemize}
Therefore, the outer AO iterates generate a (deterministically or in expectation) non-decreasing sequence
\begin{equation}
\label{outseq}
\bar R_S^{(0)} \le \bar R_S^{(1)} \le \cdots \le \bar R_S^{(t)} \le \cdots
\end{equation}
which is bounded above by the maximum achievable rate under the given power and hardware constraints. Consequently, $\{\bar R_S^{(t)}\}$ converges to a finite $\bar R_S ^\star$ with its limit point $(\mathbf v^\star,\mathbf S_{M_o}^\star)$.
\section{Complexity Analysis}
We characterize the computational complexity of the proposed AO framework for FARIS by each step on Algorithm~\ref{alg:overall_AO_compact}.
\subsection{Complexity of SAA-Based Rate Evaluation}
\label{saae}
For fixed $(\mathbf v, \mathbf S_{M_o})$, the SAA objective $\bar R_S$ in~\eqref{eq:rate_SAA} requires to compute $\{\gamma_s\}$ in~\eqref{eq:thm_gams}, which is dominated by each denominator term $\mathbf v^*\mathbf C_s\mathbf v$ with $\mathcal{O}(M_o^2)$ operations. Hence, the rate evaluation yields the complexity of $\mathcal{O}(S M_o^2)$.
\subsection{Complexity of Algorithm~\ref{alg:AO}}
\subsubsection{Variable and Constraint Dimensions}
In order to characterize the computational complexity of the subproblem in~\eqref{prob:V_SOC}, we explicitly define the dimensions of the optimization variables and the associated constraints:
\begin{equation}
\label{nms}
n_{\mathrm{var}}=M_o^2 + 2S,~n_{\mathrm{cone}}=M_o^2 + 3S,~n_{\mathrm{lin}}=M_o+1+3S,
\end{equation}
where $n_{\mathrm{var}}$ is the number of decision variables, $n_{\mathrm{cone}}$ is the total conic dimension, and $n_{\mathrm{lin}}$ is the number of linear inequality constraints.
\subsubsection{Interior-Point Complexity of Solving~\eqref{prob:V_SOC}}
A generic primal-dual interior-point step for mixed LMI/SOC problems incurs complexity of~\cite{boyd}
\begin{equation}
\label{lmisoc}
\mathcal{O}\Big(n_{\mathrm{var}}^3+n_{\mathrm{var}}^2 n_{\mathrm{cone}}+n_{\mathrm{var}} n_{\mathrm{cone}}^2+n_{\mathrm{var}}^2 n_{\mathrm{lin}}\Big).
\end{equation}
Substituting~\eqref{nms} to~\eqref{lmisoc}, the computational complexity $\text{Cost}_{\text{sub}}$ of solving~\eqref{prob:V_SOC} with $K_{\mathrm{IP}}$ steps is given by
\begin{equation}
\label{nik}
\begin{aligned}
&\mathcal{O}\Big(K_{\mathrm{IP}}\Big(
(M_o^2+2S)^3+(M_o^2+2S)^2(M_o^2+3S)\\
&~~~~~~~~~~~+(M_o^2+2S)(M_o^2+3S)^2\\
&~~~~~~~~~~~+(M_o^2+2S)^2(M_o+1+3S)
\Big)
\Big)\\
&= \mathcal{O}(K_{\mathrm{IP}}(M_o^2+S)^3).
\end{aligned}
\end{equation}
\subsubsection{Gaussian Randomization}
The Cholesky factorization of $\mathbf V_0$ requires $\mathcal{O}(M_o^3)$ flops and is performed once~\cite{chol}. For each $\mathbf z_n\sim\mathcal{CN}(\mathbf 0,\mathbf I_{M_o})$, $\mathbf v_n=\mathbf L\mathbf z_n$ is generated with complexity $\mathcal{O}(M_o^2)$. Each $\mathbf v_n$ is then projected onto the feasible set by~\eqref{eq:mag_projection_re} and~\eqref{eq:power_scaling}, dominated by latter one with $\mathcal{O}(M_o^2)$ operations. For each feasible $\mathbf v_n$, we form $\mathbf V_n=\mathbf v_n\mathbf v_n^*$ and evaluate $\tilde R_n$ by computing $\{\mathrm{tr}(\mathbf C_s\mathbf V_n),\mathrm{tr}(\mathbf a_s\mathbf a_s^*\mathbf V_n)\}_{s=1}^S$ in the objective of~\eqref{prob:Vopt_QCQP}, which requires $\mathcal{O}(S M_o^2)$ operations; hence, the cost per randomized candidate is $\mathcal{O}(S M_o^2)$. Consequently, the overall computational complexity is
\begin{equation}
\label{grcom_new}
\mathcal{O}\Big(
M_o^3
+
N_{\mathrm{rand}}S M_o^2
\Big).
\end{equation}
\subsubsection{Auxiliary Update}
Updating $y_s$ via~\eqref{eq:y_update_SOC} requires $\mathrm{tr}(\mathbf C_s\mathbf V)$ with complexity of $\mathcal{O}(SM_o^2)$.

Hence, %by adding the complexity of evaluating $\bar R_S$ by Section~\ref{saae},
the total complexity of Algorithm~\ref{alg:AO} with $T_v$ iterations is given by
\begin{equation}
\label{cost1}
\text{Cost}_{\mathbf v\text{-update}}
=\mathcal{O}\Big(T_v(K_{\mathrm{IP}}(M_o^2+S)^3+N_{\mathrm{rand}}(S M_o^2 + K_{\mathrm{IP}}' S^3))\Big).
\end{equation}
\subsection{Complexity of Algorithm~\ref{alg:CEM_ports}}
\subsubsection{Sampling and Rate Evaluation}
At iteration $t$ of Algorithm~\ref{alg:CEM_ports}, we draw $N_{\mathrm{mc}}$ i.i.d. $\{\boldsymbol{\zeta}_n^{[t]}\}$. Since generating one such vector requires $M$ independent Bernoulli trials, the total sampling complexity is $\mathcal{O}(N_{\mathrm{mc}} M)$.

For each $\boldsymbol{\zeta}_n^{[t]}$, we evaluate $\{\bar R_S(\boldsymbol\zeta_n^{[t]})\}_{n=1}^{N_{\mathrm{mc}}}$. Thus, by the result in Section~\ref{saae}, the total complexity of evaluating all $N_{\mathrm{mc}}$ samples is $\mathcal{O}(N_{\mathrm{mc}} S M_o^2)$.
\subsubsection{Elite Selection and Statistics}
Thereafter, we select the top-$\rho$ fraction of samples as $\mathcal E^{[t]}$. Sorting the $N_{\mathrm{mc}}$ fitness values requires $\mathcal{O}(N_{\mathrm{mc}}\log N_{\mathrm{mc}})$, which is optimal for comparison-based sorting. Once $\mathcal{E}^{[t]}$ of size $\lceil \rho N_{\mathrm{mc}}\rceil$ is determined, $\mu_i^{[t]}$ in~\eqref{eq:mu_def} must be computed. Each $\mu_i^{[t]}$ requires summing the $i$th entries of $N_e$ binary vectors of length $M$. Therefore, the overall complexity of computing all $M$ components of $\boldsymbol{\mu}^{[t]}$ is $\mathcal{O}(\rho N_{\mathrm{mc}} M)$.
\subsubsection{Solving for $\nu^\dagger$}
Each evaluation of $g(\nu)$ costs $\mathcal{O}(M)$, and bisection with $L_\nu$ steps gives $\mathcal{O}(L_\nu M)$.

Hence, the total complexity of Algorithm~\ref{alg:CEM_ports} with $T_{\mathrm{CEM}}$ iterations is given by
\begin{equation}
\begin{aligned}
\text{Cost}_{\mathbf S\text{-update}}=\mathcal{O}\Big(T_{\mathrm{CEM}}(&N_{\mathrm{mc}} S M_o^2+N_{\mathrm{mc}} M+N_{\mathrm{mc}}\log N_{\mathrm{mc}}\\
&+L_\nu M)\Big).
\end{aligned}
\end{equation}
\subsection{Overall AO Complexity}
Combining the results, we get the total complexity of the proposed AO framework with $T_{\mathrm{AO}}$ iterations as
\begin{equation}
\label{totcom}
\begin{aligned}
\text{Cost}_{\text{total}}=\mathcal{O}\Big(T_{\mathrm{AO}}\Big(\text{Cost}_{\mathbf v\text{-update}}+\text{Cost}_{\mathbf S\text{-update}}+SM_o^2\Big)\Big).
\end{aligned}
\end{equation}
\begin{figure}[t]
	\begin{center}
		\includegraphics[width=0.7\columnwidth,keepaspectratio]%
		{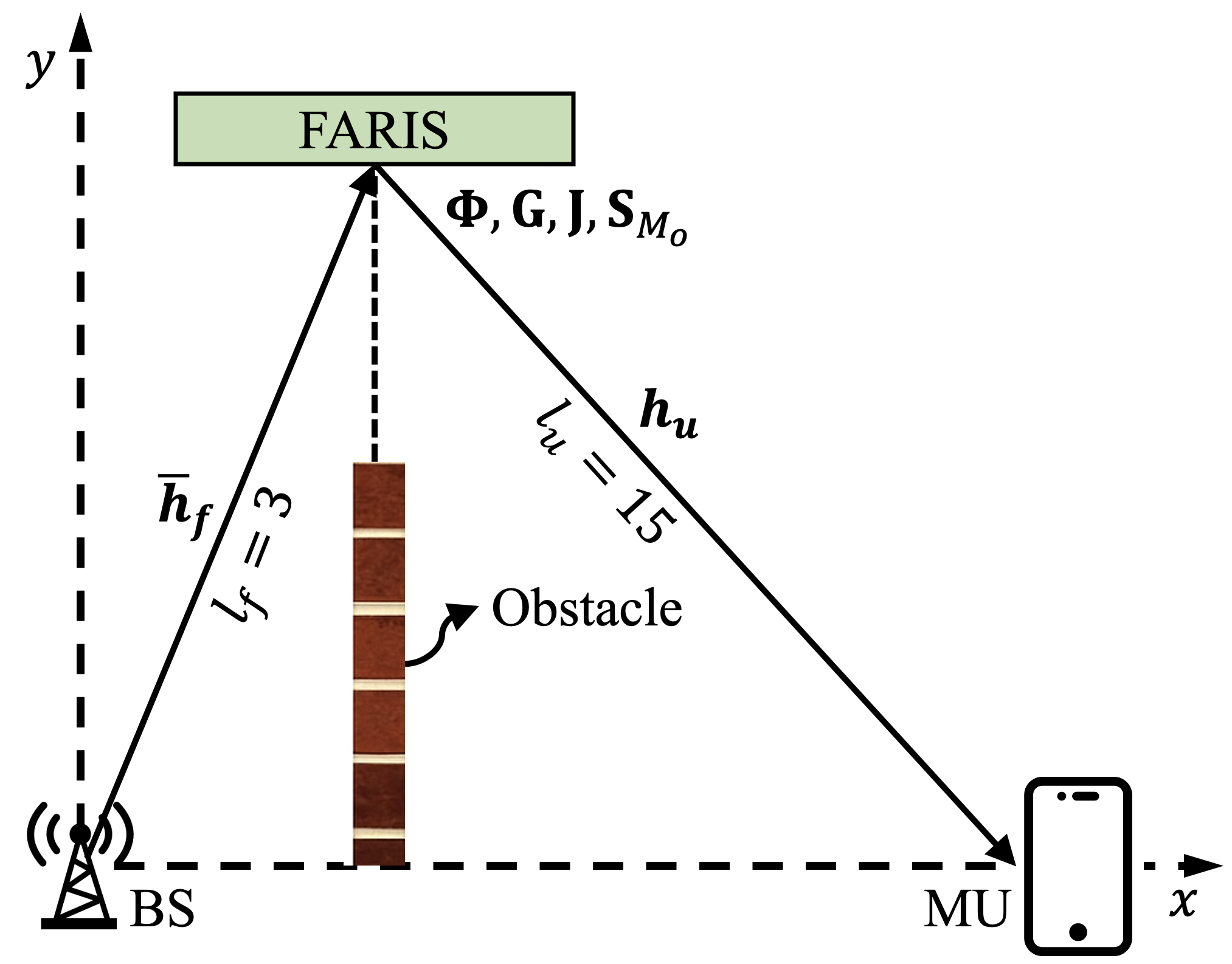}
		\caption{Simulation setup of the FARIS-aided system.}
		\label{fig_sim}
	\end{center}
\end{figure}

\section{Simulation Results}\label{simrs}
In the simulations, we assess the effectiveness of the proposed AO frameworks for maximizing $\bar{R}$ through Monte-Carlo evaluations. The FARIS setup, as shown in Fig.~\ref{fig_sim}, considers a deployment of the BS, FARIS, and MU, with distances $l_f = 3~\mathrm{m}$ and $l_u = 15~\mathrm{m}$. Since $l_f = 3~\mathrm{m}$ is smaller than $D^{\mathrm{RB}} = 6~\mathrm{m}$, the LoS condition for $\tilde{\mathbf{h}}_f$ is satisfied~\cite{nearris}. %The LoS angles are randomly drawn from $\mathcal{U}[0, \pi]$, and all 
The channels between the ARIS and MU are modeled as flat Rician fading channels, with a path-loss exponent of 2.2 and $K=1$~\cite{FASRIS22, arisover}, and the detailed simulation parameters are listed in Table~\ref{tabsim}. Each Monte-Carlo point is computed by averaging over $10^3$ independent realizations of signal transmission and reception. For comparison, we also include the performance of (i) a conventional FRIS (passive) architecture with the same number of $M_o$ FRIS elements selected from $M$ candidates with the particle-swarm-optimization (PSO)-based approach in~\cite{FRISlook}, and (ii) an ARIS architecture with same $M_o$ elements with sequential convex approximation (SCA)-based optimization in~\cite{aris5}. We additionally compare our results with a BFS baseline that exhaustively evaluates every possible port-phase configuration of the FRIS, assuming a $b=2$-bit discrete phase quantization and one-dimensional search over $[0, g_\max]$~\cite{FRISonoff}.

\begin{table}[t]
\centering
\caption{Simulation Parameters}
\label{tabsim}
\begin{tabular}{>{\centering}m {5.2cm}|>{\centering}m {2.4cm}}
\hline
\textbf{Parameter} & \textbf{Value} 
\tabularnewline
\hline
Number of FARIS/FRIS elements $M$\\(unless referred) & 100~($10\times 10$)
\tabularnewline
\hline
Maximum amplification gain $g_{\max}$ &  40 [dB]~\cite{aris5}
\tabularnewline
\hline
Hardware power consumption $(P_c, P_{\mathrm {DC}})$ & (-10, -5) [dBm]~\cite{aris5}
\tabularnewline
\hline
Total power budget of FARIS $P_{\max,t}$ & 25 [dBm]~\cite{aris5}
\tabularnewline
\hline
Transmit power $P$ (unless referred) & 15 [dBm]
\tabularnewline
\hline
Noise at FARIS and MU $(\sigma_r^2, \sigma_0^2)$ & -90 [dBm] (same)
\tabularnewline
\hline
Normalized FARIS/FRIS aperture $W_x$\\(unless referred) & 2
\tabularnewline
\hline
Parameters in CEM $(N_{\mathrm{mc}}, \rho, \omega)$ & $(5N, 0.1, 0.7)$~\cite{CEM}
\tabularnewline
\hline
\end{tabular}
\end{table}

\begin{figure*}[t]
	\begin{center}
		\includegraphics[width=1.7\columnwidth,keepaspectratio]%
		{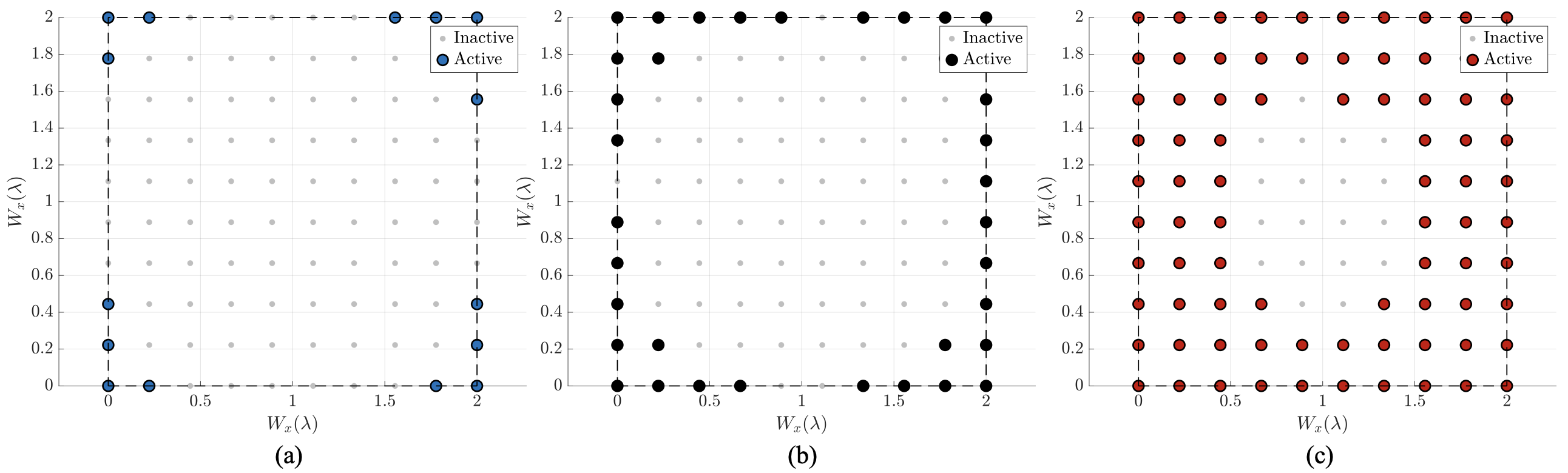}
		\caption{Illustration of examples of the optimized positions when the FARIS has (a) $M_o=16$ (b) $M_o=36$ (c) $M_o=81$ fluid active elements.}
		\label{fig_ex}
	\end{center}
\end{figure*}

Fig.~\ref{fig_ex} depicts the preset position candidates for the FARIS when employing $M_o = 16,~36$ and $81$ fluid elements, which show the inherent flexibility of the FARIS architecture. For each case, the most advantageous $\mathbf S_{M_o}^\star$ is determined by CEM-based Algorithm~\ref{alg:CEM_ports}. {The resulting element selections show a clear tendency to favor ports located near the boundary of the array rather than those concentrated around the center. This phenomenon can be interpreted from the perspective of the effective aperture and spatial diversity provided by FARIS. In particular, elements positioned farther from the array centroid typically experience larger angular separation and more diverse propagation geometries. As a result, they tend to exhibit lower spatial correlation and contribute more effectively to enlarging the achievable diversity and DoF through the optimized configuration of $\mathbf G$.}

\begin{figure}[t]
	\begin{center}
		\includegraphics[width=0.8\columnwidth,keepaspectratio]%
		{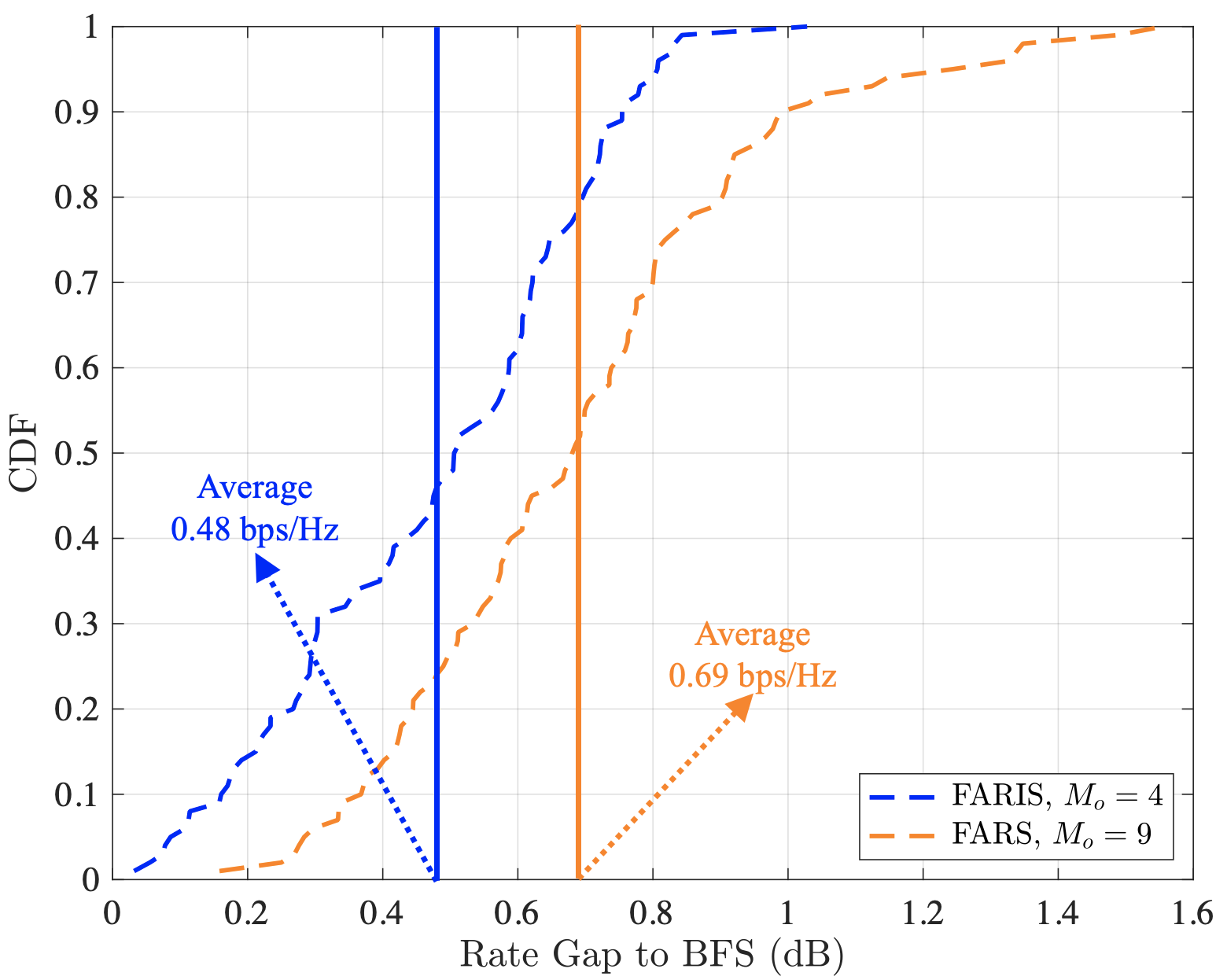}
		\caption{CDF of ergodic rate gap between BFS and proposed AO frameworks for different $M_o$ with $M=16$.}
		\label{fig_cdf}
	\end{center}
\end{figure}

Fig.~\ref{fig_cdf} depicts the cumulative distribution function (CDF) of the ergodic-rate gap, together with that of the BFS benchmark, for $M=16$ and $M_o\in\{4,9\}$. Notably, the FARIS equipped with the proposed AO framework achieves near-optimal performance, closely tracking the exhaustive BFS curve despite the dramatic reduction in computational burden. The average gap remains as small as 0.48 and 0.69~bps/Hz for $M_o=4$ and $9$, respectively, a practically negligible loss considering that BFS requires evaluating the entire discrete port-phase space and optimal amplification control, whereas our method converges rapidly with orders-of-magnitude lower complexity. This highlights the efficiency and scalability of the proposed AO framework, even under stringent configuration constraints.

\begin{figure}[t]
	\begin{center}
		\includegraphics[width=0.8\columnwidth,keepaspectratio]%
		{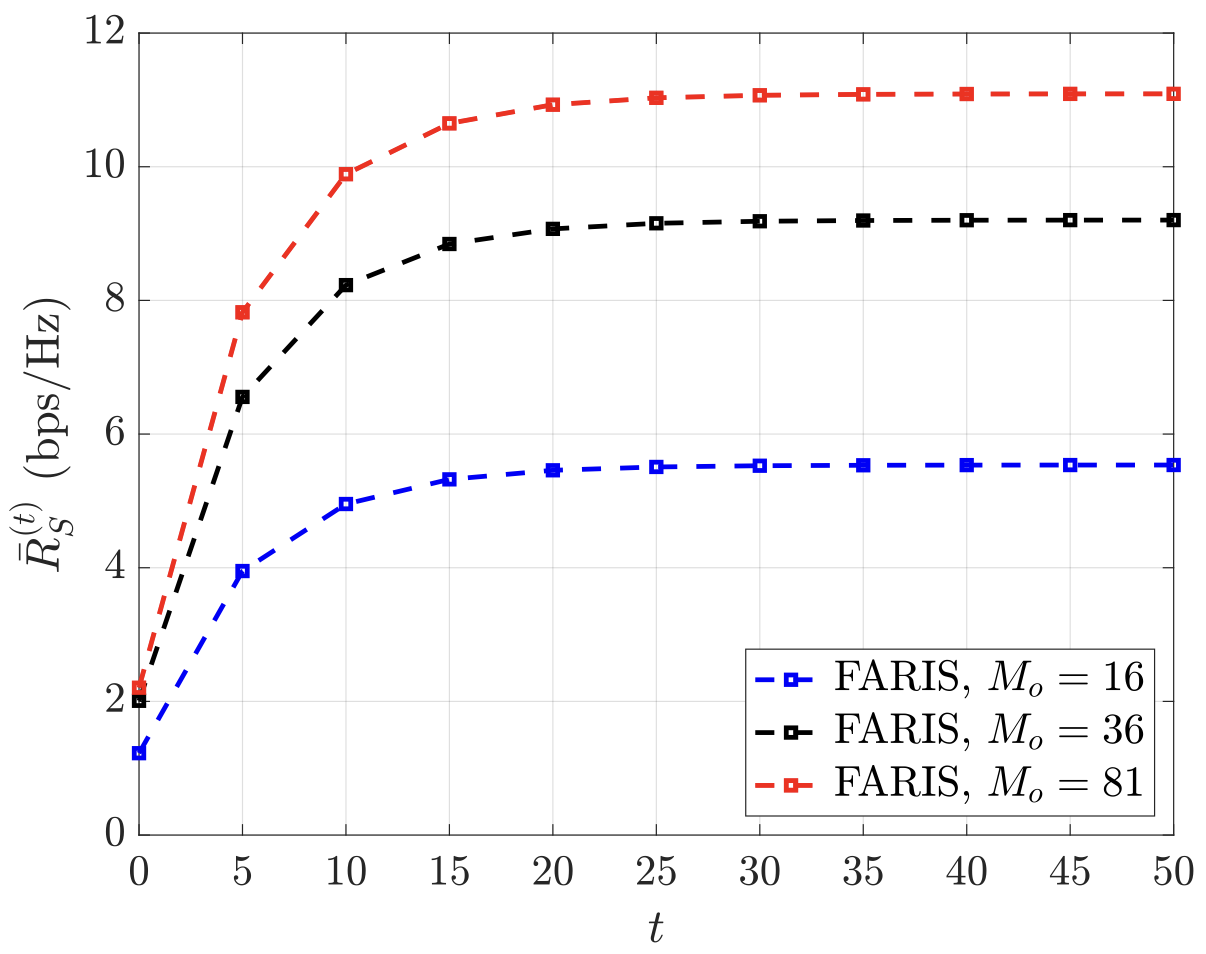}
		\caption{Convergence performance of $\bar R_S^{(t)}$ by the proposed AO framework for different $M_o$.}
		\label{fig_algo}
	\end{center}
\end{figure}
Fig.~\ref{fig_algo} illustrates the convergence behavior of $\bar R_S^{(t)}$ with the proposed AO framework for different values of $M_o$. The proposed algorithm shows a monotonic increase in the objective and typically converges within 15-20 iterations. The convergence speed is acceptable even though FARIS has more control variables (amplification + phase + port selection) than the FRIS baseline, demonstrating its practicality for real-time implementation in FARIS. Furthermore, as $M_o$ increases, convergence becomes slightly slower due to the enlarged search space and stronger spatial correlation, which lessen the marginal gain achieved at each iteration~\cite{curse}.

\begin{figure}[t]
	\begin{center}
		\includegraphics[width=0.8\columnwidth,keepaspectratio]%
		{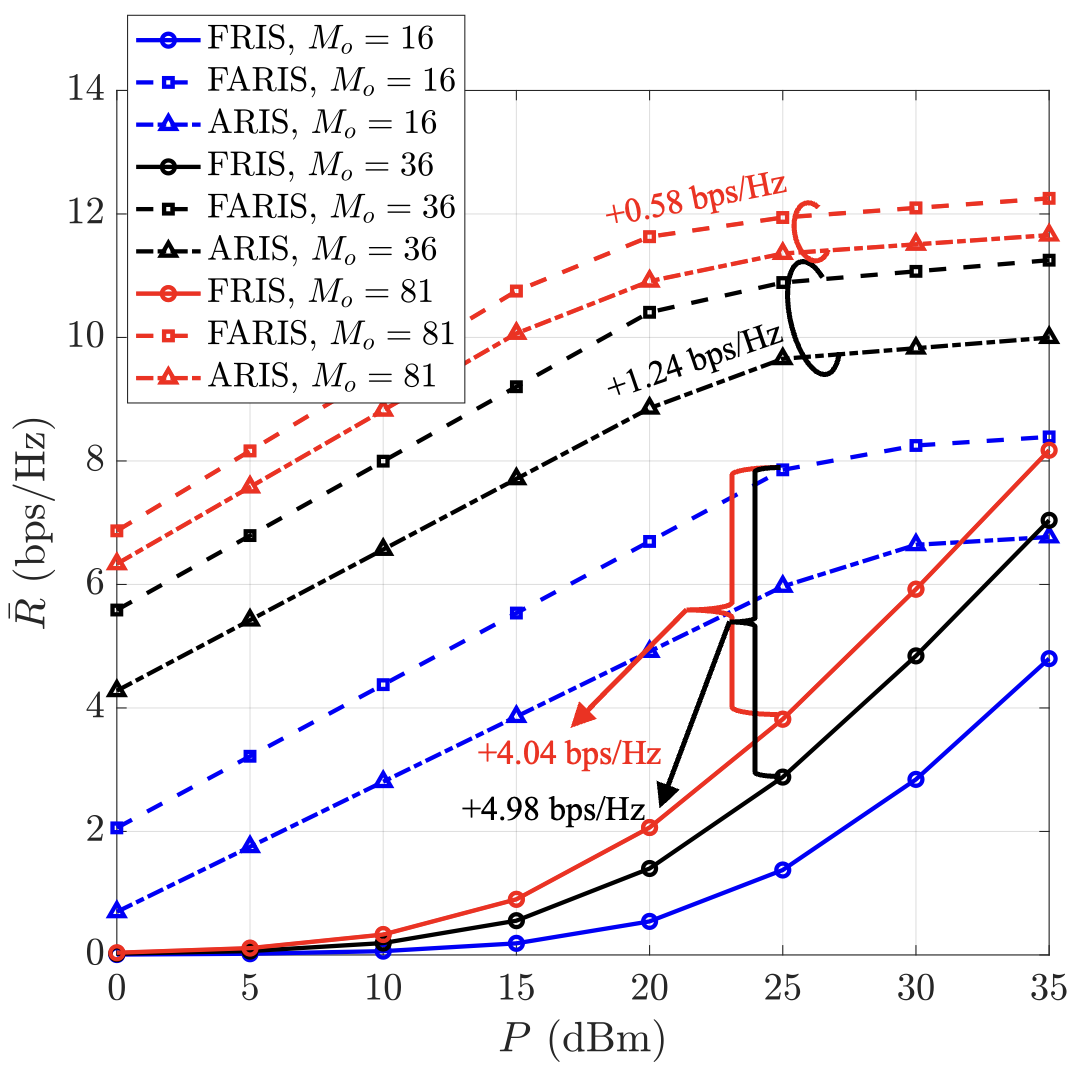}
		\caption{Achievable rate $\bar R$ versus Tx power $P$ for different $M_o$ FARIS/FRIS elements.}
		\label{fig_po}
	\end{center}
\end{figure}

Fig.~\ref{fig_po} presents $\bar R$ versus $P$ for FARIS and conventional FRIS/ARIS under $M_o=16,~36$ and $81$. As $P$ increases, FARIS with the proposed AO framework and ARIS initially exhibit a steeper rate growth; however, after certain threshold of $P$, the hardware constraint~\eqref{eq:Pris_constraint} in~\eqref{prob:maxrate} becomes active, causing both curves to transition into a more moderate scaling with respect to $P$~\cite{aris5}. Consequently, the optimization gain of FARIS over ARIS also diminishes slightly in this high-power regime, as the ARIS amplification is limited by the same hardware bound. Nevertheless, FARIS still preserves a substantial rate margin over FRIS/ARIS, even with small $M_o$, as its fluid positioning and active amplification enable more effective exploitation of the aperture and yield sustained rate gains over the counterparts; at $P=25$~dBm, the FARIS with $M_o=36$ and $81$ yields 0.36 and 0.73~bps/Hz higher rate compared to ARIS architecture. Furthermore, with same $P=25$~dBm, the FARIS configuration with $M_o=16$ attains an average rate of approximately 7.97~bps/Hz, whereas the FRIS setup, despite employing a much larger aperture, achieves no more than 3.93~bps/Hz at $M_o=81$, meaning that FARIS provides a 4.04~bps/Hz higher rate. This performance gap widens further for $M_o=36$, making the rate of FARIS 4.98~bps/Hz higher. Moreover, increasing $M_o$ improves $\bar R$ for every architecture; however, the marginal gain decreases for larger $M_o$ with fixed $M$, reflecting a saturation due to dense spatial sampling shown in Fig.~\ref{fig_ex}.
\begin{figure}[t]
	\begin{center}
		\includegraphics[width=0.8\columnwidth,keepaspectratio]%
		{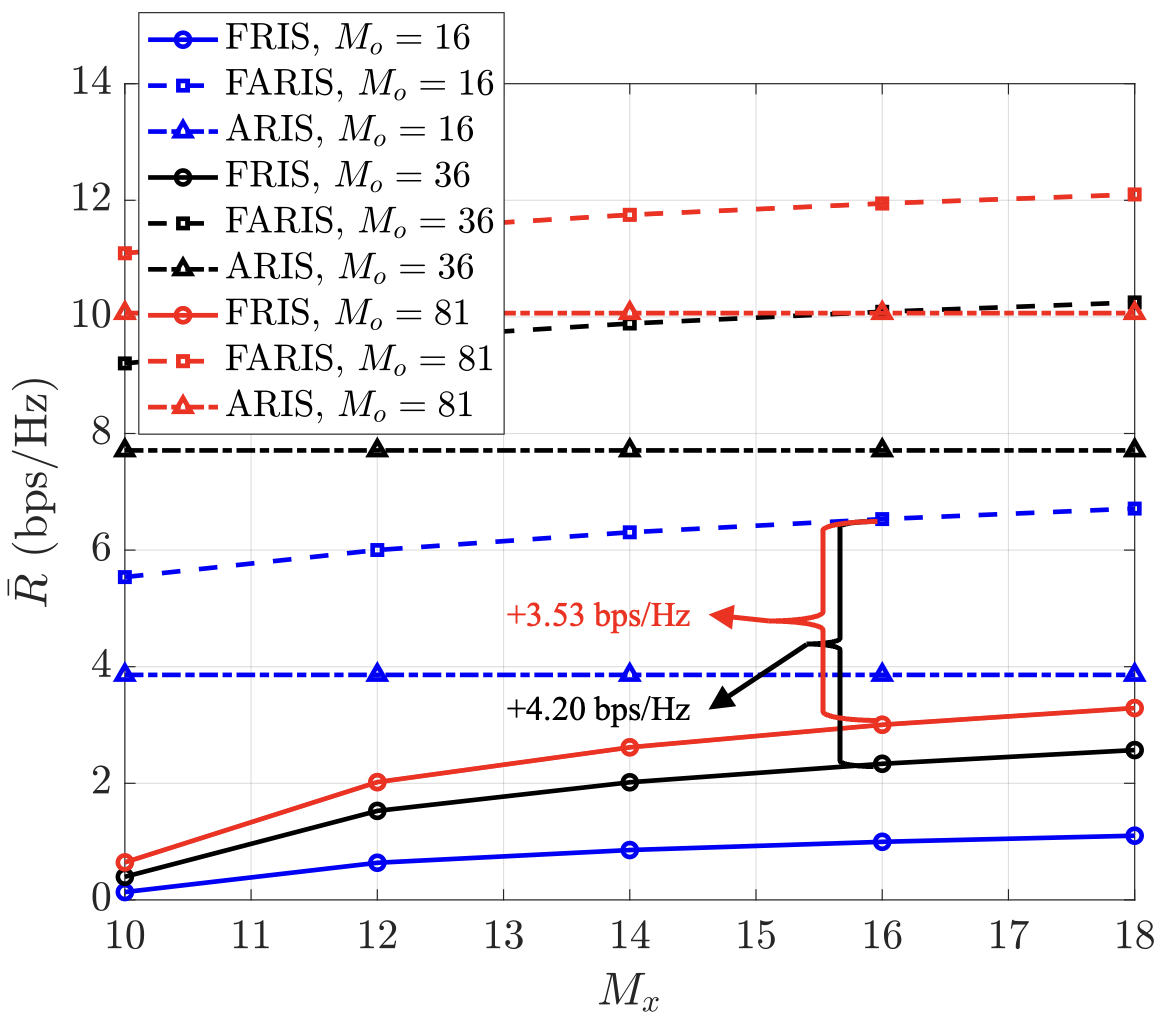}
		\caption{Achievable rate $\bar R$ versus the number of total elements $M~(M_x\times M_x)$ for different $M_o$ FARIS/FRIS elements.}
		\label{fig_mx}
	\end{center}
\end{figure}

Fig.~\ref{fig_mx} depicts the variation of $\bar R$ with respect to $M$ (or equivalently $M_x$) in the FARIS and conventional FRIS/ARIS configuration. For any fixed M, FARIS equipped with the proposed AO framework consistently achieves a higher rate than FRIS/ARIS, even with less $M_o$; when $M_x = 16$, the proposed FARIS configuration with $M_o = 16$ provides an additional 3.53~bps/Hz over FRIS with $M_o=81$, which increases to 4.20~bps/Hz when $M_o=36$, and the gain over ARIS, whose fixed architecture yields a nearly flat rate profile due to the absence of extra DoFs or adaptive optimization, becomes readily apparent in the figure. This again highlights the substantial rate gain enabled by both fluid repositioning and active amplification. As $M$ increases, both FARIS/FRIS architectures benefit from the denser aperture, which enhances spatial sampling and increases the likelihood of selecting more favorable element positions. Despite this common scaling benefit, FARIS maintains a clear performance margin over FRIS, as its fluid port placement together with active amplification enables more effective utilization of the expanded aperture, thereby providing rate enhancements unattainable by the purely passive counterpart.
\begin{figure}[t]
	\begin{center}
		\includegraphics[width=0.8\columnwidth,keepaspectratio]%
		{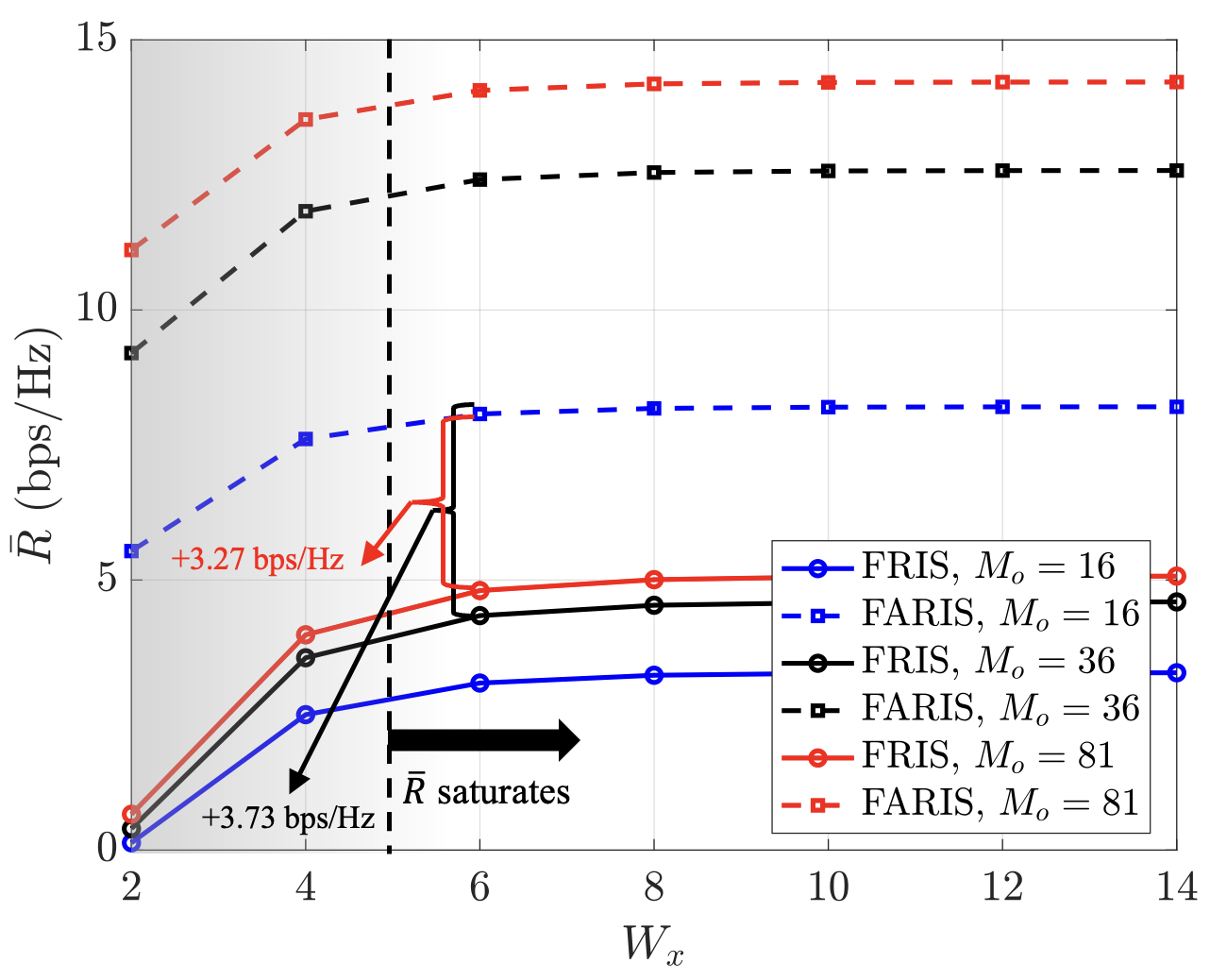}
		\caption{Achievable rate $\bar R$ versus normalized aperture $W_x$ for different $M_o$ FARIS/FRIS elements.}
		\label{fig_wx}
	\end{center}
\end{figure}

Fig.~\ref{fig_wx} illustrates $\bar R$ as a function of $W_x$ for both FARIS and conventional FRIS. In the initial growth region, all schemes, including the proposed AO-enabled FARIS and the conventional FRIS, exhibit a sharp increase in $\bar R$ due to the rapid reduction in spatial correlation and the improved ability to select favorable port positions as the aperture expands. This effect is also considerably pronounced for FARIS, similar to the previous observation, so that even a FARIS design with a smaller element count (e.g., $M_o=16$) can surpass the rate of a larger-aperture FRIS configuration (e.g., $M_o=36,~81$). As $W_x\gtrapprox5$, which corresponds to spacing of $\frac{\lambda}{2}$, the rate improvement begins to saturate since the additional aperture provides only marginal decorrelation benefits~\cite{perlimfas}. Notably, in this region, FARIS exploits the enlarged aperture far more effectively than FRIS, thanks to its jointly optimized fluid positioning and active amplification. Therein with $W_x=6$, the rate difference between FARIS with $M_o=16$ and FRIS when $M_o=81$ is 3.27~bps/Hz, which further widens to 3.73~bps/Hz when the number of active FRIS elements is reduced to $M_o = 36$. Although all schemes eventually show saturation after $W_x\gtrapprox 5$, FARIS maintains a consistently larger gain throughout the entire range of $W_x$.

\begin{figure}[t]
	\begin{center}
		\includegraphics[width=0.8\columnwidth,keepaspectratio]%
		{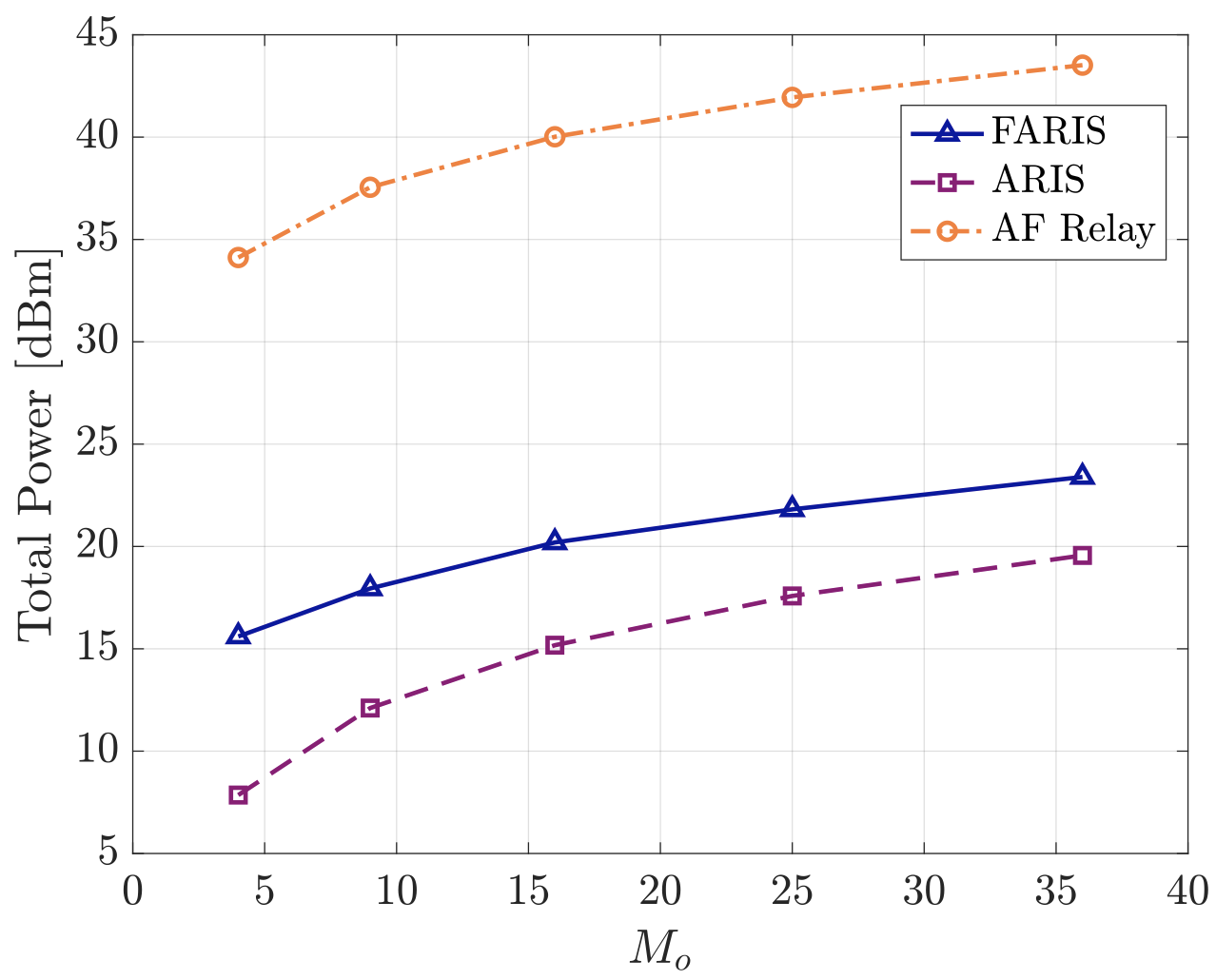}
		\caption{Total power consumption versus $M_o$ for the proposed FARIS and conventional ARIS and AF relay architectures.}
		\label{fig_pocon}
	\end{center}
\end{figure}

%Fig.~\ref{fig_pocon} presents the total power consumption as a function of $M_o$ for both FARIS and conventional ARIS. As observed, the total power consumption increases monotonically with $M_o$ in both architectures, since selecting a larger number of elements requires additional $P_{\mathrm{DC}}$ and results in a higher $P_{\mathrm{RIS}}$. In contrast to FARIS, the ARIS architecture consistently achieves lower total power consumption across the considered range of $M_o$. This behavior can be attributed to the structural difference between the two architectures. In ARIS, the active elements are directly deployed at the selected $M_o$ locations, leading to a smaller hardware power term $M_o(P_c+P_{\mathrm{DC}})$ as well as a reduced $P_{\mathrm{RIS}}$. On the other hand, FARIS maintains a larger candidate pool with $M$ ports, resulting in a higher hardware power term $MP_c+M_oP_{\mathrm{DC}}$ in~\eqref{eq:Pfaris_total}, and the fluid amplification mechanism can also increase the reflected power component, which was shown in previous figures. Consequently, the total power consumption of FARIS becomes higher than that of ARIS. Nevertheless, the additional power overhead remains moderate, and when combined with the substantial rate improvement provided by FARIS, as illustrated in the previous figures, the overall FARIS architecture remains practically feasible. These results indicate that FARIS achieves significant performance gains in achievable rate while incurring only a limited increase in power consumption compared to ARIS.

{Fig.~\ref{fig_pocon} illustrates the total power consumption versus $M_o$ for the proposed FARIS architecture, conventional ARIS, and AF relay. Herein, the hardware power consumption of the $M_o$-element AF relay is modeled as~\cite{cute}.
\begin{equation}
\label{afpower}
P_{\mathrm{hard,AF}}=M_o(P_{\mathrm{DAC}} + P_{\mathrm{mix}} + P_{\mathrm{filt}}) + P_{\mathrm{syn}},
\end{equation}
where $P_{\mathrm{DAC}}, P_{\mathrm{mix}}, P_{\mathrm{filt}}$, and $P_{\mathrm{syn}}$ denote the power consumption of the digital-to-analog converter (DAC), mixer, filter, and frequency synthesizer, respectively. According to the practical parameters reported in~\cite{cute}, the combined circuit power $P_{\mathrm{DAC}} + P_{\mathrm{mix}} + P_{\mathrm{filt}}$ is approximately $18.2$~dBm, while $P_{\mathrm{syn}}=17$~dBm. This hardware power is then combined with the communication-related power for AF relays~\cite{af1, af3}. As shown in the figure, compared with FARIS, the ARIS architecture consistently exhibits lower total power consumption across the considered range of $M_o$. This is mainly due to the structural difference between the two architectures. In ARIS, the active elements are directly deployed at the selected $M_o$ positions, resulting in a smaller hardware power term $M_o(P_c+P_{\mathrm{DC}})$ as well as a reduced reflected signal power $P_{\mathrm{RIS}}$. In contrast, FARIS maintains a larger candidate pool with $M$ ports, which introduces an additional control-circuit power $MP_c$ in~\eqref{eq:Pfaris_total}. %Moreover, the fluid amplification mechanism can further increase the reflected signal power component. As a result, FARIS consumes slightly more power than ARIS. 
Nevertheless, the power consumption of FARIS remains significantly lower than that of the AF relay. This is because unlike the AF relay nodes that operate as independent transceivers with dedicated radio-frequency (RF) chains and signal forwarding operations~\cite{af1, af3}, whose hardware power in~\eqref{afpower} introduces substantially higher total power, FARIS remains a metasurface-based structure in Fig.~\ref{fig_sche} that directly manipulates the impinging electromagnetic wave over a distributed aperture. Therefore, FARIS provides a favorable intermediate design that achieves notable rate improvements over ARIS while avoiding the excessive power overhead associated with AF relay architectures.}

\section{Conclusion}
%In this paper, we proposed the FARIS concept, a novel fluid-active-RIS architecture that jointly exploits fluid port repositioning and per-port amplification to substantially enhance the performance of RIS-aided wireless network. We developed a correlated FARIS channel model based on Jakes’ spatial correlation and formulated an ergodic-rate maximization problem under a realistic reflection-power constraint, jointly optimizing active reflection coefficients and port selection. To tackle the resulting mixed-integer non-convex program, we introduced an AO framework that alternates between: i) a lifted active-reflection design solved via SAA and a quadratic-transform-based conic optimization in the lifted matrix variable, and ii) an $M_o$-constrained CEM-based port-selection scheme with a closed-form characterization guaranteeing a unique solution and monotonic improvement of the underlying likelihood. Simulation results showed that FARIS with proposed AO framework converges within a modest number of iterations and consistently outperforms a conventional FRIS/ARIS baseline across a wide range of transmit powers, total elements, and apertures, often achieving higher rates even with fewer active elements or smaller apertures. These findings demonstrate that combining fluid repositioning with active amplification provides a powerful additional DoF over FRIS/ARIS designs, positioning FARIS as a promising building block for 6G wireless environments.
In this paper, we proposed the FARIS concept, a novel fluid-active-RIS architecture that jointly exploits fluid port repositioning and per-port amplification to substantially enhance the performance of RIS-aided wireless networks. To capture the practical hardware behavior, we developed a circuit-aware modeling framework and established a correlated FARIS signal model based on Jakes’ spatial correlation. Under a realistic reflection-power constraint, we formulated an ergodic-rate maximization problem that jointly optimizes the active amplification-reflection coefficients and the discrete selection of fluid-active ports. To solve the resulting mixed-integer non-convex problem, we developed an AO-based framework that efficiently alternates between active-reflection optimization and CEM-based port selection, achieving near-optimal performance with significantly reduced computational complexity. Numerical results demonstrate that the proposed FARIS architecture consistently outperforms conventional baselines, achieving higher ergodic rates even with fewer active elements than the benchmark schemes, thereby highlighting the efficiency of the proposed fluid-active design. These results highlight that the synergy between fluid reconfigurability and active amplification introduces a powerful additional spatial DoF, enabling more efficient channel shaping and positioning FARIS as a promising architectural paradigm for 6G.

\section*{Appendix A}
\section*{Proof of Theorem~\ref{thm:reformulation}}
We first prove~\eqref{eq:thm_gams}. From~\eqref{eq:Aact} and~\eqref{eq:thm_buk},
\begin{equation}
\label{eq:thm_signal1}
\mathbf h^{(s)*}\mathbf A_{\mathrm{act}}\tilde{\mathbf h}_f=\mathbf u_s^* \mathrm{diag}(\mathbf v)\mathbf b,
\end{equation}
and the square of magnitude becomes
\begin{equation}
\label{eq:thm_signal2}
\big|\mathbf h^{(s)*}\mathbf A_{\mathrm{act}}\tilde{\mathbf h}_f\big|^2=\mathbf v^* \mathbf a_s \mathbf a_s^* \mathbf v=|\mathbf a_s^*\mathbf v|^2,
\end{equation}
which is the numerator of $\gamma_s$ in~\eqref{eq:rate_SAA}. Thereafter, the ARIS-induced noise term satisfies
\begin{equation}
\begin{aligned}
\mathbf h^{(s)*}\mathbf A_{\mathrm{act}}\mathbf A_{\mathrm{act}}^*\mathbf h^{(s)}&= \mathbf u_s^*\mathrm{diag}(\mathbf v)\mathbf K \mathrm{diag}({\mathbf v^*})\mathbf u_s\\
&=\mathbf v^*\big(\mathbf K \odot (\mathbf u_s\mathbf u_s^*)\big)\mathbf v,
\label{eq:noise_step2}
\end{aligned}
\end{equation}
Thus by~\eqref{eq:thm_signal2} and~\eqref{eq:noise_step2}, the denominator of $\gamma_s$ in~\eqref{eq:rate_SAA} is
\begin{equation}
\label{eq:thm_den}
\sigma_0^2 + \mathbf v^*\mathbf C_s\mathbf v~(\mathbf C_s \triangleq L_u\sigma_r^2(\mathbf K \odot (\mathbf u_s\mathbf u_s^*))),
\end{equation}
and using $\mathbf V=\mathbf v\mathbf v^*$ and $\mathrm{tr}(\mathbf X\mathbf v\mathbf v^*)=\mathbf v^*\mathbf X\mathbf v$, we obtain~\eqref{eq:thm_gams}.

We now prove~\eqref{eq:thm_Pris}. Again by using~\eqref{eq:Aact} and~\eqref{eq:thm_buk},
\begin{equation}
\label{eq:thm_Q1step}
\|\mathbf A_{\mathrm{act}}\tilde{\mathbf h}_f\|_2^2=\mathbf v^*\mathbf Q_1\mathbf v,~\mathbf Q_1
\triangleq\mathrm{diag}(\mathbf b^*)\mathbf K\mathrm{diag}(\mathbf b)=\mathbf K \odot (\overline{\mathbf b}\mathbf b^{\mathrm T}).
\end{equation}
Similarly,
\begin{equation}
\label{eq:thm_Q2}
\begin{aligned}
\mathrm{tr}\left(\mathbf A_{\mathrm{act}}\mathbf A_{\mathrm{act}}^*\right)&=\mathrm{tr}(\mathrm{diag}(\mathbf v^*)\mathbf K\mathrm{diag}(\mathbf v)\mathbf K)\\
&=\mathbf v^*\mathbf Q_2\mathbf v~(\mathbf Q_2 \triangleq \mathbf K\odot \mathbf K).
\end{aligned}
\end{equation}
Therefore, by defining $\mathbf F \triangleq P L_f\mathbf Q_1 + \sigma_r^2\mathbf Q_2$, $P_{\mathrm{RIS}}$ becomes
\begin{equation}
\label{eq:thm_totalvv}
P_{\mathrm{RIS}}=\mathbf v^*\mathbf F\mathbf v=\mathrm{tr}(\mathbf F\mathbf V),
\end{equation}
and the theorem follows. $\blacksquare$
\section*{Appendix B}
\section*{Proof of Theorem~\ref{th1}}
At iteration $t$ of Algorithm~\ref{alg:AO}, we assume that $(\mathbf V^{[t]},\{y_s^{[t]}\}_{s=1}^S)$ has been obtained. By construction of the update rule in~\eqref{eq:y_update_SOC}, the equality holds in the quadratic-transform lower-bound in~\eqref{eq:transformed_obj} at $\mathbf V^{[t]}$, i.e.,
\begin{equation}
\label{rsxi}
\bar R_S^{[t]}
%=
%\frac{1}{S}\sum_{s=1}^S \log_2\bigl(1+\gamma_s^{[t]}\bigr)
=
\frac{1}{S}\sum_{s=1}^S \log_2\bigl(1+\xi_s^{[t]}\bigr),
\end{equation}
%where $\gamma_s^{[t]}$ and $\xi_s^{[t]}$ denotes $\gamma_s$ and $\xi_s$ evaluated at $\mathbf V^{[t]}$ and $(\mathbf V^{[t]},\{y_s^{[t]}\})$, respectively.
where $\xi_s^{[t]}$ denotes $\xi_s$ evaluated at $(\mathbf V^{[t]},\{y_s^{[t]}\})$. Fixing $\{y_s^{[t]}\}$, we solve~\eqref{prob:V_SOC} to obtain $\mathbf V^{[t+1]}$. Since $\mathbf V^{[t+1]}$ maximizes the lower-bound in~\eqref{eq:transformed_obj} for fixed $\{y_s^{[t]}\}$, it holds that
\begin{equation}
\label{eq:qt_step1_new}
\frac{1}{S}\sum_{s=1}^S \log_2\bigl(1+\bar{\xi}_s^{[t+1]}\bigr)\ge\frac{1}{S}\sum_{s=1}^S \log_2\bigl(1+\xi_s^{[t]}\bigr),
\end{equation}
where $\bar\xi_s^{[t+1]}$ denotes $\xi_s$ evaluated at $(\mathbf V^{[t+1]},\{y_s^{[t]}\})$. Then $\{y_s^{[t+1]}\}$ are updated via~\eqref{eq:y_update_SOC}, which is the unique maximizer of the quadratic-transform identity in~\eqref{eq:quad_transform_identity}. Therefore, the equality holds in the lower-bound in~\eqref{eq:transformed_obj} at $\mathbf V^{[t+1]}$, i.e.,
\begin{equation}
\label{rxsii}
\bar R_S^{[t+1]}=\frac{1}{S}\sum_{s=1}^S 
\log_2\bigl(1+\xi_s^{[t+1]}\bigr)\ge\frac{1}{S}\sum_{s=1}^S \log_2\bigl(1+\bar{\xi}_s^{[t+1]}\bigr).
\end{equation}
Hence, by combining~\eqref{rsxi}-\eqref{rxsii}, we obtain
\begin{equation}
\label{monor_new}
\bar R_S^{[t+1]}
=
\frac{1}{S}\sum_{s=1}^S \log_2\bigl(1+\xi_s^{[t+1]}\bigr)
\ge
\frac{1}{S}\sum_{s=1}^S \log_2\bigl(1+\xi_s^{[t]}\bigr)
=
\bar R_S^{[t]}.
\end{equation}
Thus, $\{\bar R_S^{[t]}\}$ monotonically non-decreases, and the theorem follows. $\blacksquare$
\section*{Appendix C}
\section*{Proof of Theorem~\ref{lem:CEM_KKT}}
The Lagrangian of~\eqref{prob:CE_main} is
\begin{equation}
\label{lag}
\mathcal L(\mathbf p,\nu)=\sum_i\big(\mu_i^{[t]}\log p_i+(1-\mu_i^{[t]})\log(1-p_i)\big)+\nu\big(M_o-\sum_i p_i\big).
\end{equation}
Stationarity yields $\frac{\mu_i^{[t]}}{p_i}-\frac{1-\mu_i^{[t]}}{1-p_i}-\nu=0$, which rearranges to~\eqref{eq:quad_eq}. For $\nu=0$, $p_i=\mu_i^{[t]}$ holds. For $\nu\neq0$, the quadratic condition in~\eqref{eq:quad_eq} admits two real roots since $\mu_i^{[t]}\in(0, 1)$:
\begin{equation}
\label{eq:roots_app}
p_i^{(\pm)}(\nu)=\frac{(\nu+1)\pm\sqrt{(\nu+1)^2-4\nu\mu_i^{[t]}}}{2\nu},
\end{equation}
which implies the following two scenarios:
\begin{enumerate}
\item $\nu>0$: The product of $p_i^{(\pm)}(\nu)$ becomes $\frac{\mu_i^{[t]}}{\nu}>0$, so the roots have same signs, and since $f(0, \nu)=\mu_i^{[t]}>0$ and $f(1, \nu)=\mu_i^{[t]}-1<0$, one root lies in $(0,1)$ and the other in $(1,\infty)$. Because the denominator $2\nu>0$, $p_i^{(+)}>p_i^{(-)}$ holds. Therefore, $p_i^{(+)}>1$ while $p_i^{(-)}\in(0,1)$.
\item $\nu<0$: The product of $p_i^{(\pm)}(\nu)$ becomes $\frac{\mu_i^{[t]}}{\nu}<0$, so the roots have opposite signs. Again, $f(0, \nu)=\mu_i^{[t]}>0$ and $f(1, \nu)=\mu_i^{[t]}-1<0$ imply that the positive root lies in $(0,1)$ and the negative one is less than 0. Since the denominator $2\nu<0$, $p_i^{(-)}>p_i^{(+)}$ holds. Therefore, $p_i^{(-)}\in(0,1)$ still holds.
\end{enumerate}
Hence, $p_i^{(-)}(\nu)\in(0, 1)$ holds for all $\nu\neq0$, and we denote it simply by $p_i(\nu)$ as in~\eqref{eq:p_of_nu}.

To prove strictly decreasing, define $f(p_i,\nu)\triangleq\nu p_i^2 - (\nu + 1)p_i + \mu_i^{[t]} = 0$. By the implicit function theorem:
\begin{equation}
\frac{\partial p_i}{\partial \nu}
= -\frac{\frac{\partial f_i}{\partial \nu}}{\frac{\partial f_i}{\partial p_i}}
= -\frac{p_i^2 - p_i}{2\nu p_i - (\nu + 1)}.
\label{eq:dpdnu_app}
\end{equation}
Because $p_i\in(0,1)$, we have $p_i^2 - p_i = p_i(p_i - 1)<0$, so the numerator in~\eqref{eq:dpdnu_app} is negative. For denominator, by letting $\Delta\triangleq(\nu+1)^2-4\nu\mu_i^{[t]}$, from~\eqref{eq:p_of_nu} we have $2\nu p_i-(\nu+1) = -\sqrt{\Delta}<0$. Combining these signs gives
\begin{equation}
\label{ctsg}
\frac{\partial p_i}{\partial \nu}
= -\frac{\text{(negative)}}{\text{(negative)}} < 0,
\end{equation}
which proves that $p_i(\nu)$ is strictly decreasing in~$\nu$.

Continuity of $p_i(\nu)$ follows directly from its closed-form expression in~\eqref{eq:p_of_nu}, which implies that $g(\nu)\triangleq\sum_i p_i(\nu)$ is also strictly decreasing and continuous. Moreover, as $\nu\to -\infty$ leads to $p_i(\nu)\to 1$ so $g(\nu)\to M$, and as $\nu\to +\infty$, $p_i(\nu)\to 0$, so $g(\nu)\to 0$. Hence by the intermediate value theorem, for any target $M_o\in(0,M)$, there exists a unique $\nu^\dagger$ satisfying $g(\nu^\dagger)=M_o$, and the theorem follows. $\blacksquare$

\bibliographystyle{IEEEtran}
\bibliography{IEEEexample}

\end{document}